\documentclass[lettersize,journal]{IEEEtran}
\usepackage{amsmath,amsfonts}
\usepackage{algorithmic}
\usepackage{algorithm}
\usepackage{array}
\usepackage[caption=false,font=normalsize,labelfont=sf,textfont=sf]{subfig}
\usepackage{textcomp}
\usepackage{stfloats}
\usepackage{url}
\usepackage{verbatim}
\usepackage{graphicx}
\usepackage{amsthm}
\usepackage{soul}
\usepackage{cite}
\usepackage{color}
\usepackage{epstopdf}

\newtheorem{proposition}{Proposition}

\usepackage{amssymb}

\newif\ifshowchanges
 \showchangesfalse

\ifshowchanges
  \newcommand{\change}[1]{\textcolor{blue}{#1}}
\else
  \newcommand{\change}[1]{#1}
\fi

\begin{document}

\title{Two-Phase Channel Estimation for RIS-Assisted THz Systems with Beam Split}


\author{Xin~Su, Ruisi~He,~\IEEEmembership{Senior Member,~IEEE}, Peng~Zhang, Bo~Ai,~\IEEEmembership{Fellow,~IEEE}, Yong~Niu,~\IEEEmembership{Senior Member,~IEEE} and Gongpu~Wang,~\IEEEmembership{Member,~IEEE}
		\thanks{Part of this paper was published at the IEEE VTC 2024 Spring.}
		\thanks{X. Su, R. He, P. Zhang, B. Ai and Y. Niu are with the State Key Laboratory of Advanced Rail Autonomous Operation, the School of Electronics and Information Engineering, and the Frontiers Science Center for Smart High-speed Railway System, Beijing Jiaotong University, Beijing 100044, China (e-mail: 22120124@bjtu.edu.cn; ruisi.he@bjtu.edu.cn; 21120170@bjtu.edu.cn; boai@bjtu.edu.cn; niuy11@163.com).}
	\thanks{G. Wang is with the Beijing Key Laboratory of Transportation Data Analysis and Mining, School of Computer and Information Technology, Beijing Jiaotong University, Beijing 100044, China, and also with the China (Wuxi) Institute of Internet of Things, Wuxi 214111, China (e-mail: gpwang@bjtu.edu.cn).}

	\thanks{This work is supported by the Fundamental Research Funds for the Central Universities under Grant 2022JBQY004, and the National Natural Science Foundation of China under Grant 62431003 and 62271037.}

}

\maketitle

\begin{abstract}
Reconfigurable intelligent surface (RIS)-assisted terahertz (THz) communication is emerging as a key technology to support the ultra-high data rates in future sixth-generation networks. However, the acquisition of accurate channel state information (CSI) in such systems is challenging due to the passive nature of RIS and the hybrid beamforming architecture typically employed in THz systems. To address these challenges, we propose a novel low-complexity two-phase channel estimation scheme for RIS-assisted THz systems with beam split effect. In the proposed scheme, we first estimate the full CSI over a small subset of subcarriers (SCs), then extract angular information at both the base station and RIS. Subsequently, we recover the full CSI across remaining SCs by determining the corresponding spatial directions and angle-excluded coefficients. Theoretical analysis and simulation results demonstrate that the proposed method achieves superior performance in terms of normalized mean-square error while significantly reducing computational complexity compared to existing algorithms.

\end{abstract}

\begin{IEEEkeywords}
Terahertz, beam split, reconfigurable intelligent surface, wideband channel estimation, direction estimation
\end{IEEEkeywords}

\section{Introduction}
\IEEEPARstart{E}{nvisioned} as a cornerstone technology to meet the requirements of future 6G networks, terahertz (THz) communication holds immense potential for a myriad of applications within the field of the Internet of Things, ranging from intelligent industrial manufacturing to digital twinning technologies and immersive experiences like holographic telepresence and the metaverse \cite{song2011present,qin2023multi,niu2024solid}. However, THz communication systems encounter various challenges, including high propagation loss, limited coverage, and susceptibility to atmospheric absorption and scattering \cite{zhang2023coupling,chen2021intelligent}. One notable drawback is the severe signal attenuation experienced by transmitted signals, which often necessitate a line-of-sight (LoS) link to maintain communication quality \cite{yuan20223d}. To address these challenges and unlock the full potential of THz communication, reconfigurable intelligent surface (RIS) has emerged to provide virtual LoS path by intelligently manipulating the propagation environment \cite{zhang2023joint}. RIS can dynamically adjust reflection properties of a large number of passive reflecting elements to optimize signal transmission and reception \cite{bjornson2022reconfigurable}. By strategically controlling the phase shifts of these reflecting elements, RIS can effectively enhance the signal strength, mitigate interference, and improve the overall performance of THz communication systems. Despite these benefits, the RIS elements remain passive and lack inherent signal processing capabilities. Accordingly, the performance enhancement enabled by RIS relies significantly on the accurate channel state information (CSI), which is challenging to acquire. Therefore, investigating an efficient channel estimation scheme for RIS-assisted THz communication systems is of great importance.
\subsection{Related Works}
Over the past few years, numerous research works have been dedicated to exploring channel estimation in RIS-assisted high frequency systems \cite{wang2020compressed,xwei2021channel,ardah2021trice,zhou2022channel,ning2021terahertz,jian2020modified}. Specifically, a sparse representation of RIS-assisted millimeter wave (mmWave) channel has been established, enabling the application of compressed sensing based techniques for CSI acquisition \cite{wang2020compressed}. In \cite{xwei2021channel}, a low-overhead channel estimation method has been proposed based on the double-structure characteristic of RIS-assisted angular cascaded channel. To reduce computational complexity, \cite{ardah2021trice} proposes a non-iterative two-stage framework based on direction-of-arrival (DoA) estimation. The authors in \cite{zhou2022channel} propose a two-phase channel estimation strategy to reduce pilot overhead. \cite{ning2021terahertz} investigates cooperative-beam-training based channel estimation procedure. In \cite{jian2020modified}, a sparse Bayesian learning based channel estimation approach has been developed to compensate the severe off-grid issue. 

The above solutions have effectively addressed the channel estimation problem of narrowband systems with utilization of the angular domain sparsity. However, mmWave/THz communications typically adopt a hybrid beamforming architecture to reduce hardware costs\cite{ning2023beamforming}, making beam squint/beam split effect inevitable. These effects occur due to the wider bandwidth of mmWave/THz signals and the enormous number of RIS elements, which causes path components to split into distinct spatial directions at different subcarrier (SC) frequencies \cite{tan2019delay}. Such phenomenon significantly degrades achievable data rates, thereby counteracting the benefits of RIS-assisted THz systems.

To address this issue, beam squint/beam split effect in RIS-assisted mmWave/THz systems has been investigated in \cite{elbir2023bsa,dovelos2021channel,su2024channel,wu2023parametric}. In particular, \cite{elbir2023bsa} introduces a beam-split-aware orthogonal matching pursuit (BSA-OMP) approach to automatically mitigate beam split effect. In \cite{dovelos2021channel}, a low-complexity beam squint mitigation scheme is proposed based on generalized-OMP estimator. The authors in \cite{su2024channel} propose a robust regularized sensing-beam-split-OMP-based scheme to estimate beam split affected RIS-assisted cascaded channel. The authors in \cite{wu2023parametric} regard the wideband channel estimation problem as a block-sparse recovery problem by concatenating signals at different SCs. These works primarily focus on mitigating beam split effect within the context of cascaded channels using well-designed dictionaries covering the entire beamspace. Although they can achieve accurate channel estimation, it requires execution at each SC or simultaneous processing of all SCs, which may introduce additional unaffordable computations.

An alternative approach is to focus on the physical directions of the channel components to reduce computational complexity \cite{ozen2024beam,ma2021wideband,liu2021cascaded,tan2021wideband,wei2022accurate}. In mmWave band, angle-based channel estimation methods are used to simplify the wideband channel estimation process under the beam squint effect \cite{ozen2024beam,ma2021wideband}. The authors in \cite{ozen2024beam} first estimate angle-of-arrivals (AoAs) at BS and then employed BSA-based algorithms to acquire full CSI. \cite{ma2021wideband} presents a twin-stage OMP method to estimate the path angles of the RIS-assisted cascaded two-hop channel. \cite{liu2021cascaded} introduces a Newtonized OMP (NOMP) based algorithm to detect the channel parameters. \change{These approaches perform angle estimation by leveraging the principle that the frequency-dependent spatial angles change while frequency-independent physical angles remain constant across different SCs. This principle can also be applied in THz systems, however, due to the significantly larger beam spacing in THz band compared to mmWave systems, the frequency-dependent spatial angles exhibit greater deviations.
Moreover, the solutions discussed above fail to leverage the wider bandwidth available in THz band to better differentiate between true physical and false spatial angles.
Consequently, the angle-based methods used in mmWave systems may become neither efficient nor effective when applied to the THz band. }

To this end, several low-complexity estimation schemes have been developed specifically for wideband THz systems \cite{tan2021wideband,wei2022accurate}. For example, \cite{tan2021wideband} proposes a beam-split-pattern detection-based estimation method that prioritizes physical DoA determination. In \cite{wei2022accurate}, a three-stage wideband channel estimation approach is presented, where AoAs and angles-of-departure (AoDs) are successively refined to reconstruct the channel under beam split effect. These methods achieve a favorable trade-off between complexity and performance in traditional architectures. \change{However, these methods are not readily applicable to RIS-assisted THz systems due to the lack of mechanisms for angular estimation at RIS. Moreover, the passive nature of RIS further complicates the accurate acquisition of angular information, making it difficult to resolve the coupled AoAs/AoDs.} Therefore, this motivates the development of a low-complexity channel estimation scheme in RIS-assisted wideband THz systems with beam split effect.
\begin{figure}[!t]
	\centering
	\includegraphics[width=0.48\textwidth]{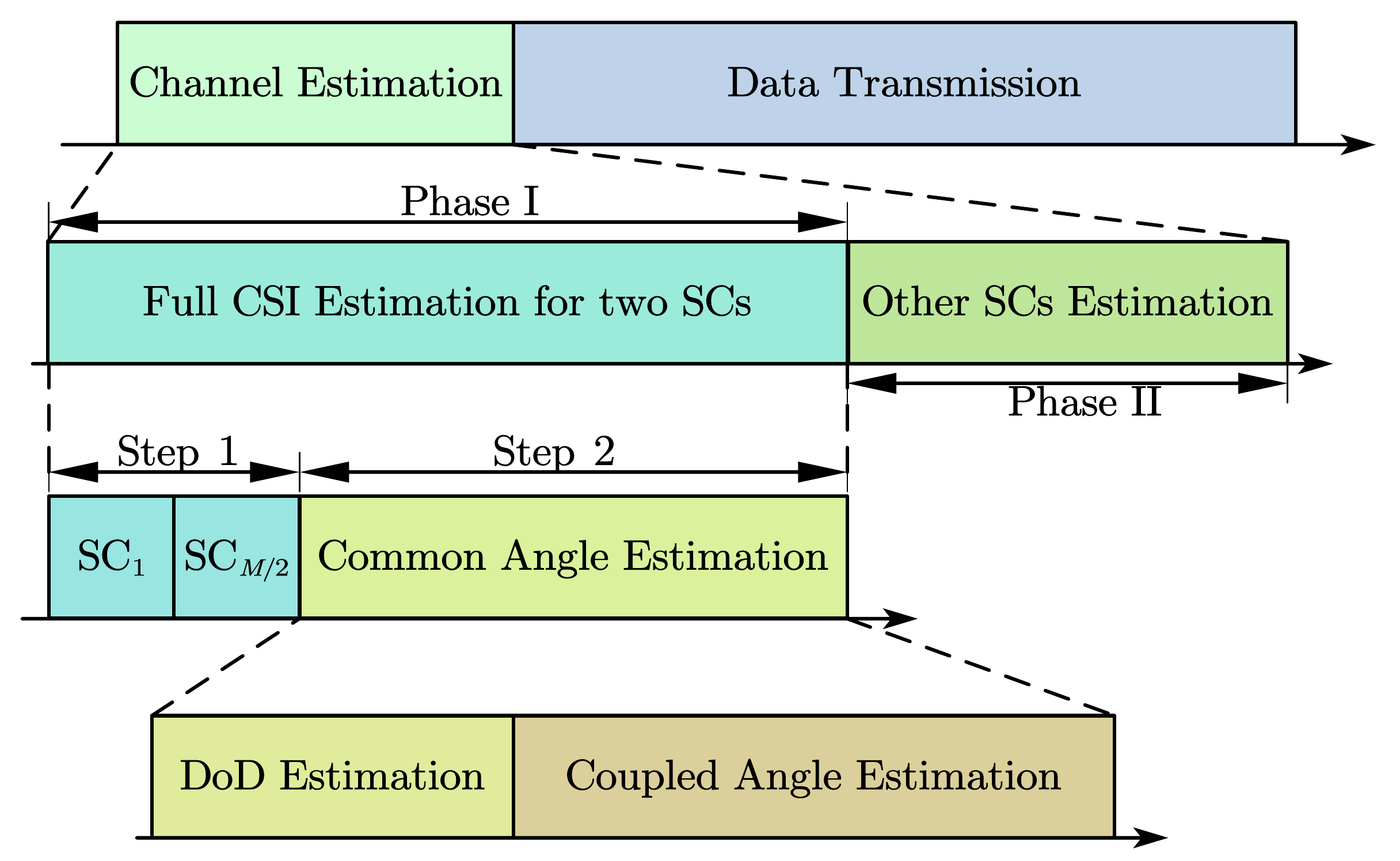}
	\caption{Illustration of the two-phase channel estimation protocol and corresponding frame structure.}
	\label{fig:frame_structure}
\end{figure}
\subsection{Main Contributions}
To address the issues discussed above, we presented some preliminary results on estimating the RIS-assisted cascaded channel \cite{su2024generalized}. In this paper, we significantly expand on the previous work and propose a two-phase efficient channel estimation scheme for RIS-assisted wideband THz systems with beam split. Our main contributions are summarized as~follows:
\begin{itemize}
    \item \change{We propose a novel low-complexity channel estimation scheme tailored for RIS-assisted wideband THz communication systems under the beam split effect. Unlike conventional methods that rely on full-band sparse recovery across all subcarriers, the proposed scheme adopts a two-phase protocol to achieve accurate CSI estimation with substantially reduced computations. 
    Specifically, as illustrated in Fig.~\ref{fig:frame_structure}, Phase I estimates the full CSI at two carefully selected SCs, while  Phase II efficiently reconstructs the remaining CSI using only a simple least squares (LS) algorithm, thereby avoiding the need of computationally expensive recovery procedures across the full bandwidth.}

    \item \change{The proposed low-complexity design is primarily enabled by the two-step approach implemented in Phase~I. In Step~1, we recover the full CSI at two selected SCs utilizing a simplified cyclic beam split (CBS) dictionary, which is compressed from the original size of $Q^2$ to a linear scale of $2Q-1$, thereby significantly reducing the complexity of sparse recovery. In Step~2, we estimate the physical directions at both the BS and RIS based on the recovered CSI. Specifically, we propose an energy maximum (EnM) algorithm for directions-of-departure (DoDs) estimation at BS and a double sensing-MUSIC (DS-MUSIC) algorithm to resolve the coupled angles at RIS. Instead of mitigating the severe beam misalignment introduced by beam split effect, the proposed algorithm leverages the resulting spatial separation of RIS angles within the extended angular range, which facilitates accurate estimation of physical directions. 
      } 
    
    \item \change{We provide a comprehensive complexity analysis, which demonstrates that our proposed protocol significantly reduces computational complexity compared to existing sparse recovery solutions. Furthermore, we conduct extensive numerical evaluations  under both mmWave and THz system settings. The results demonstrate that: (i) the proposed two-phase scheme efficiently exploits the beam split properties to enhance estimation accuracy compared to mmWave solutions; and (ii) it achieves superior NMSE performance over conventional baselines, approaching that of the oracle-LS benchmark. These findings highlight the potential of our method for efficient channel estimation in RIS-assisted THz communication systems.}

\end{itemize}

The remainder of this paper is organized as follows. Section \ref{sec:System Model} presents the system model for RIS-assisted wideband THz communication with beam split effect along with the associated channel estimation protocol. Section \ref{sec:full_CSI_estimation}, \ref{sec:angle_estimation} and \ref{sec:CE_other_SC} describe the proposed two-phase channel estimation framework in detail. Section~\ref{sec:complexity_analysis} analyzes the computational complexity, and Section~\ref{sec:Numerical Results and Discussion} provides simulation results and performance comparisons with existing methods. Finally, Section~\ref{sec:Conclusion} concludes the paper.

{\it{Notations}}: Boldface lower-case and capital letters represent the column vectors and matrices, respectively. $\left(\cdot \right)^*$, $(\cdot )^{\mathsf{T}}$, $\left(\cdot \right)^{\mathsf{H}}$ and $\left(\cdot \right)^{\dagger}$ denote the conjugate, transpose, transpose-conjugate and pseudo-inverse operation, respectively. $\mathbb{E}$ and $\mathrm{Var}$ represent the statistical expectation and variance, respectively. $\mathbb{N}$, $\mathbb{C}$ and $\mathbb{R}$ denote the set of integers, complex numbers and real numbers, respectively. $\mathbf{I}_{N}$ is the $N\times N$ size identity matrix and $\mathrm{tr}$ is the trace. $\otimes $, $\odot$ and $\bullet$ denote the Kronecker product, Khatri-Rao product and transposed Khatri-Rao product, respectively. $[\mathbf{A}]_{i,j}$ represents the $(i,j)$-th element of matrix $\mathbf{A}$, and $[\mathbf{a}]_i$ represents the $i$-th element of vector or set $ \mathbf{a} $. Besides, the operator $\text{Diag}(\mathbf{x})$ constructs a diagonal matrix whose diagonal elements are given by the elements of the vector $\mathbf{x}$. The notation $a:b$ denotes the sequence $\{a, a+1, \dots, b\}$ with $a \leq b$. The imaginary unit is denoted as  $\jmath = \sqrt{-1}$.

\section{System and Channel Model}\label{sec:System Model}
\subsection{Signal and Channel Model}
As depicted in Fig. \ref{fig:system_model}, we consider a RIS-assisted wideband THz MIMO system employing hybrid beamforming architecture that integrates both analog and baseband processing. Assuming that the BS equipped with $N_{\mathrm{T}}$ antennas and $N_\mathrm{RF}$ (where $N_\mathrm{RF}<N_{\mathrm{T}}$) RF chains is deployed to serve $K$ single-antenna UEs. The RIS employs a uniform linear array (ULA) comprised of $N_{\mathrm{R}}$ reflecting elements, which is manipulated by a smart controller operated by the BS. Besides, the orthogonal frequency division multiplexing (OFDM) system is considered with $M$ SCs, central frequency $f_c$ and bandwidth $B$. 
\begin{figure}[!ht]
	\centering
	\includegraphics[width=0.48\textwidth]{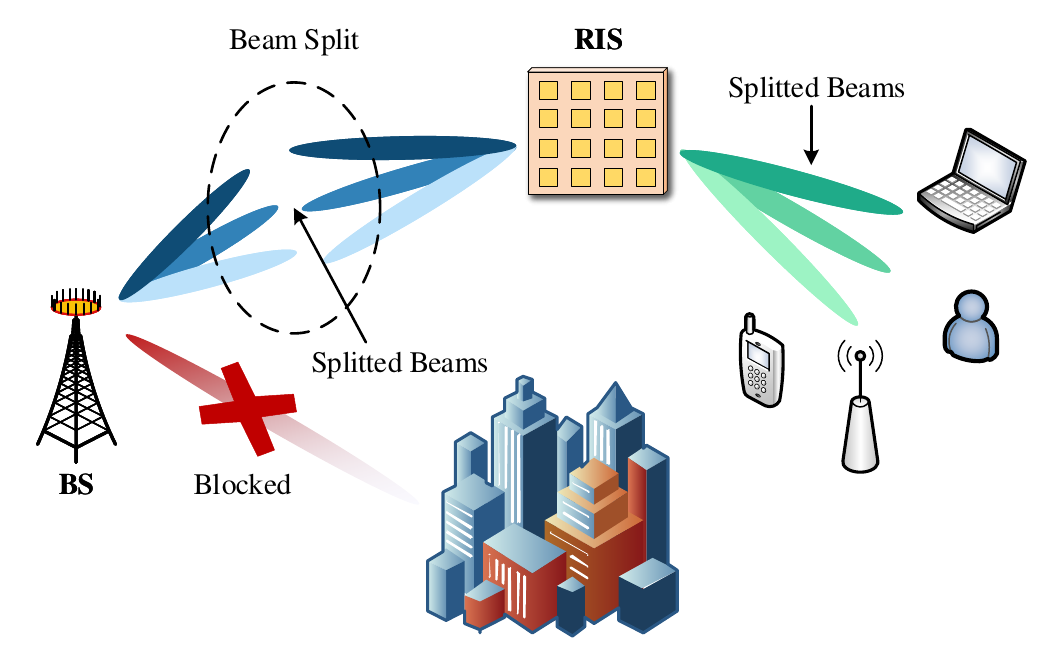}
	\caption{Illustration of RIS-assisted THz communication system with beam split effect.}
	\label{fig:system_model}
\end{figure}

In this paper, we assume that the channel is quasi-static, i.e., the channel remains approximately constant in a channel coherence block with $D$ subframes \cite{zhou2022channel}. For each coherence block, channel estimation is first performed during $T$ subframes, and downlink data transmission is then carried out during the remaining $D-T$ subframes. Moreover, each subframe of the downlink channel estimation period is divided into $P$ time slots, with $P$ baseband beamformers operating alternatively across different time slots.
Additionally, we assume that the direct BS-UE links are blocked by obstacles \cite{wu2023parametric}\footnote{The proposed scheme can be easily extended to one scenario where direct BS-UE and cascaded BS-RIS-UE links coexist. The direct channel can be estimated via conventional channel estimation approaches by turning off RIS.}. Therefore, we mainly focus on the downlink channel estimation of the cascaded BS-RIS-UE channel.

As a cost-performance trade-off, hybrid beamforming architecture is typically employed in THz MIMO systems\cite{ning2023beamforming}. During $P$ successive time slots, the frequency-dependent baseband beamformer vectors $\left\{ \mathbf{w}_{t,p}^{\mathrm{BB}}[m] \right\} _{p=1}^{P}\in \mathbb{C} ^{N_{\mathrm{RF}}}$ are sequentially employed combining with the frequency-independent analog beamformer $\mathbf{W}^{\mathrm{RF}}\in \mathbb{C} ^{N_{\mathrm{T}}\times N_\mathrm{RF}}$. Each entry of $\mathbf{W}^{\mathrm{RF}}$ is subject to the unit-modulus constraint, i.e., $|[\mathbf{W}^{\mathrm{RF}} ]_{i,j}|=1/\sqrt{N_{\mathrm{T}}} $ for $i=1, \ldots, N_{\mathrm{T}}$ and $j=1, \ldots, N_{\mathrm{RF}}$. We assume that the analog beamformer remains unchanged over multiple channel coherence blocks due to the relatively static nature of the RIS physical location compared to the dynamic channel coefficients \cite{zhou2022channel}\footnote{Note that the angular information remains unknown for channel estimation.}. Therefore, we define $\mathbf{w}_{t,p}[m]\triangleq \mathbf{W}^{\mathrm{RF}}\mathbf{w}_{t,p}^{\mathrm{BB}}[m]\in \mathbb{C} ^{N_{\mathrm{T}}}$ as the hybrid beamformer vector. Let $\mathbf{G}\left[ m \right]\in \mathbb{C} ^{N_{\mathrm{R}}\times N_{\mathrm{T}}}$ and $\mathbf{h}_k\left[ m \right] \in \mathbb{C} ^{N_{\mathrm{R}}}$ denote the channel from BS to RIS and RIS to $k$-th UE at the $m$-th SC, respectively. At time slot $p$ of subframe $t$, the received measurement $y_{k,t,p}[m] $ by $k$-th UE at the $m$-th SC can be expressed as  
\begin{equation}\label{equ:y_p}
		\\
			y_{k,t,p}[m]=\!\boldsymbol{\theta }_t\mathrm{Diag}\left( \mathbf{h}_k\left[ m \right] \right)\! \mathbf{G}\left[ m \right] \mathbf{w}_{t,p}[m]s_{t}[m]\!+\!n_{k,t,p}[m],
\end{equation}
where $n_{k,t,p}[m]\sim \mathcal{CN} \left( 0,2\varepsilon\right) $ is the complex additive white Gaussian noise (AWGN). Besides, $\boldsymbol{\theta }_t = \left[ e^{\jmath \theta_{t,1}}, \cdots, e^{\jmath \theta_{t,N_{\mathrm{R}}}}\right]$ is the phase-shift matrix of RIS at subframe $t$ with a constant reflection amplitude constraint \cite{yang20161}. 

Let $\mathbf{r}_{t,p}[m]\triangleq \mathbf{w}_{t,p}[m]s_{t}[m]\in \mathbb{C} ^{N_{\mathrm{T}}}$ denote the processed transmitted signal and combining the processed signals during $P$ successive time slots $\mathbf{R}_t[m]\triangleq \left [\mathbf{r}_{t,1}[m],...,\mathbf{r}_{t,P}[m] \right] \in \mathbb{C} ^{N_{\mathrm{T} }\times P}$, the combined measurement vector $\tilde{\mathbf{y}}_{k,t}[m] \in \mathbb{C}^{1 \times P}$ is expressed as 
\begin{equation}\label{equ:y}
\tilde{\mathbf{y}}_{k,t}[m]=\boldsymbol{\theta }_t\mathrm{Diag}\left( \mathbf{h}_k\left[ m \right] \right)\! \mathbf{G}\left[ m \right]\mathbf{R}_t[m]+\tilde{\mathbf {n}}_{k,t}[m],
\end{equation}
where $\tilde{\mathbf{n}}_{k,t}[m]\triangleq \left[ n_{k,t,1}[m],...,n_{k,t,P}[m] \right] \in \mathbb{C} ^{1\times P}$ denotes the combined AWGN. Therefore, the cascaded channels can be expressed as \cite{yu2019enabling}
\begin{equation}\label{equ:CasCh}
	\mathbf{H}_{k}^{\mathrm{cas}}\left[ m \right] =\mathrm{Diag}\left( \mathbf{h}_k\left[ m \right] \right)  \mathbf{G}\left[ m \right] \in \mathbb{C}^{N_{\mathrm{R}}\times N_{\mathrm{T}}}.
\end{equation}


We adopt the widely used Saleh-Valenzuela channel model for both BS-RIS channel and RIS-UE channel \cite{tse2005fundamentals}. The BS-RIS channel $\mathbf{G}[m]$ at the $m$-th SC is expressed as 
\begin{equation}\label{equ:G}
	\mathbf{G}[m]=\sum_{l=1}^L{\alpha _l\mathbf{a}_{\mathrm{R},m}( \tilde{\psi}_l ) \mathbf{a}_{\mathrm{T},m}^{\mathsf{H}}\left( \tilde{\varphi}_l \right) e^{-\jmath2\pi \tau _lf_m}},
\end{equation}
where $L$ represents the number of propagation paths between BS and RIS, $\alpha _{l}$, $\tau _{l}$, $\tilde{\psi}_l$ and $\tilde{\varphi}_l$ represent the complex gain, the time delay, the physical DoA at BS and DoD at RIS for the $l$-th path, respectively. Besides, $f_m$ denotes the frequency of the $m$-th SC, which is given by 
\begin{equation}\label{equ:f_m}
	f_m=f_c+\frac{B}{M}\left(m-1-\frac{M-1}{2}\right),\forall m=1,2, \dots, M,
\end{equation}
where $f_c$ is the central carrier frequency. Besides, $\mathbf{a}_{\mathrm{T},m}\in \mathbb{C} ^{N_{\mathrm{T}}}$ and $\mathbf{a}_{\mathrm{R},m}\in \mathbb{C} ^{N_{\mathrm{R}}}$ are the normalized array response vectors (ARVs) associated with BS and RIS, respectively, given by 
\begin{subequations}
	\begin{equation}\label{equ:at}
		\mathbf{a}_{\mathrm{T},m}(\tilde{u} ) \triangleq \frac{1}{\sqrt{N_{\mathrm{T}}}}e^{-\jmath2\pi \sin \left( \tilde{u}\right) \mathbf{n}_td_{\mathrm{BS}}/\lambda_m},
	\end{equation}
	\begin{equation}\label{equ:ar}
		\mathbf{a}_{\mathrm{R},m}(\tilde{v}) \triangleq \frac{1}{\sqrt{N_{\mathrm{R}}}}e^{-\jmath2\pi \sin \left( \tilde{v}  \right) \mathbf{n}_rd^{\mathrm{RIS}}/\lambda_m},
	\end{equation}
\end{subequations}
where $\lambda_m=c/f_m$ is the wavelength of the $m$-th SC. $d_{\mathrm{BS}}$ and $d^{\mathrm{RIS}}$ denote the antenna spacing of BS and RIS, respectively. Besides, $\mathbf{n}_t=[ 0,1,\dots ,N_{\mathrm{T}}-1 ]^{\mathsf{T}}$ and $\mathbf{n}_r=\left[ 0,1,\dots ,N_{\mathrm{R}}-1 \right]^{\mathsf{T}}$, respectively. By defining the angle vectors $\tilde{\boldsymbol{\psi}}=[ \tilde{\psi}_1,\cdots ,\tilde{\psi}_L ]$ and $\tilde{\boldsymbol{\varphi}}=\left[ \tilde{\varphi}_1,\cdots ,\tilde{\varphi}_L \right] $, $\mathbf{G}[m]$ can be rewritten as 
\begin{equation}\label{equ:VAD_G}
\mathbf{G}[m]=\mathbf{A}_{\mathrm{R},m}( \tilde{\boldsymbol{\psi}} ) \mathbf{\Sigma }_m\mathbf{A}_{\mathrm{T},m}^{\mathsf{H}}\left( \tilde{\boldsymbol{\varphi}} \right),
\end{equation}
where $\mathbf{A}_{\mathrm{R},m}( \tilde{\boldsymbol{\psi}}) =[ \mathbf{a}_{\mathrm{R},m}(\tilde{\psi}_1),\cdots ,\mathbf{a}_{\mathrm{R},m}(\tilde{\psi}_L) ]$ and $\mathbf{A}_{\mathrm{T},m}\left( \tilde{\boldsymbol{\varphi}} \right) =\left[ \mathbf{a}_{\mathrm{T},m}(\tilde{\varphi}_1),\cdots ,\mathbf{a}_{\mathrm{T},m}(\tilde{\varphi}_L) \right]$ denote the array response matrix (ARM) of the DoDs at BS and DoAs at RIS, respectively. $\mathbf{\Sigma }_m=\mathrm{Diag}\left( \alpha _1e^{-\jmath2\pi \tau _1f_m},\cdots ,\alpha _Le^{-\jmath2\pi \tau _Lf_m} \right) $ is the angle-excluded matrix. 

Similarly, the channel $\mathbf{h}_k[m]  \in \mathbb{C}^{1 \times N_{\mathrm{R}}}$ between RIS and $k$-th UE at the $m$-th SC  can be formulated as 
\begin{equation}\label{equ:u}
\mathbf{h}_k[m]=\sum_{j=1}^{J_k}{\beta _{k,j}\mathbf{a}_{\mathrm{R},m}^{\mathsf{H}}( \tilde{\vartheta}_{k,j}) e^{-\jmath2\pi \mu _{k,j}f_m}},
\end{equation}
where $J_k$ represents the number of propagation paths between RIS and $k$-th UE, $\beta _{k,j}$, $\mu _{k,j}$ and $\tilde{\vartheta}_{k,j}$ represent the complex gain, the time delay and the DoD at RIS for the $j$-th path, respectively. By defining $\tilde{\boldsymbol{\vartheta}}_k=[ \tilde{\vartheta}_{k,1},\cdots ,\tilde{\vartheta}_{k,J_k} ]$, $\mathbf{h}_k[m]$ can be rewritten as
\begin{equation}\label{equ:VAD_u}
	\mathbf{h}_k[m]=\boldsymbol{\beta }_{k,m}^{\mathsf{T}}\mathbf{A}_{\mathrm{R},m}^{\mathsf{H}}( \tilde{\boldsymbol{\vartheta}}_k),
\end{equation}
where $\boldsymbol{\beta }_{k,m}\!=\!\left[ \beta _{k,1}e^{-\jmath2\pi \mu _{k,1}f_m},\cdots\! ,\beta _{k,J_k}e^{-\jmath2\pi \mu _{k,J_k}f_m} \right] ^{\mathsf{T}} $ is the angle-excluded vector for the $k$-th UE. $\mathbf{A}_{\mathrm{R},m}( \tilde{\boldsymbol{\vartheta}}_k )=[ \mathbf{a}_{\mathrm{R},m}(\tilde{\vartheta}_{k,1}),\cdots ,\mathbf{a}_{\mathrm{R},m}(\tilde{\vartheta}_{k,J_k}) ] $ denotes the ARM of the DoDs at RIS. 
Substituting \eqref{equ:VAD_G} and \eqref{equ:VAD_u} into \eqref{equ:CasCh}, the cascaded channel $\mathbf{H}_{k}^{\mathrm{cas}}[m]$ can be rewritten as
\begin{equation}\label{equ:VAD_H}
			\mathbf{H}_{k}^{\mathrm{cas}}[m]=( \mathbf{A}_{\mathrm{R},m}^{*}( \tilde{\boldsymbol{\vartheta}}_k ) \!\bullet\! \mathbf{A}_{\mathrm{R},m}( \tilde{\boldsymbol{\psi}} ) )( \boldsymbol{\beta }_{k,m}\!\otimes \!\mathbf{\Sigma }_m ) \mathbf{A}_{\mathrm{T},m}^{\mathsf{H}}( \tilde{\boldsymbol{\varphi}} ).
\end{equation}


\subsection{Beam Split Effect}
According to \eqref{equ:G} and \eqref{equ:u}, THz channel varies across different SCs due to the beam split effect\footnote{The beam split effect can been seen as a severe version of beam squint effect in mmWave band \cite{su2023wideband}.}. This phenomenon is caused by the ultra-wide bandwidth and extensive antenna arrays and reflecting elements in RIS-assisted wideband THz systems. The beams generated by the frequency-independent analog beamforming vector and then reflected by the frequency-independent RIS phase shift matrix would split into totally distinct spatial directions.
Typically, antenna spacings are set to half the wavelength at the central frequency $f_c$, i.e., $d_\mathrm{BS}=d_\mathrm{RIS}=\lambda_c/2$, where $\lambda_c = c/f_c$ is the wavelength of the central frequency \cite{dai2022delay}. For simplicity, we define the variables $u \triangleq \sin\tilde{u} \in [-1,1]$ and $v\triangleq\sin\tilde{v}\in [-1,1]$ to represent the physical directions. Substituting the half-wavelength configuration into \eqref{equ:at} and \eqref{equ:ar}, the ARVs can be rewritten as
\begin{subequations}\label{equ:bs_arv}
	\begin{equation}
		\mathbf{a}_{\mathrm{T},m}(u )=\frac{1}{\sqrt{N_{\mathrm{T}}}}e^{-\jmath\pi \eta _m u \mathbf{n}_t},
	\end{equation}
	\begin{equation}
		\mathbf{a}_{\mathrm{R},m}(v)=\frac{1}{\sqrt{N_{\mathrm{R}}}}e^{-\jmath\pi \eta _m v \mathbf{n}_r},
	\end{equation}
\end{subequations}
where $\eta_m\triangleq f_m / f_c$ represents the relative frequency.
As shown in \eqref{equ:bs_arv}, the spatial directions can be represented by $u_m=\eta_mu$ and $v_m=\eta_mv$.
\change{In the lower frequency bands, such as mmWave, the SC frequencies are relatively closer, resulting in only minor deviations of $\eta_m$ from one. Consequently, beam squint remains limited, and the ARVs exhibit only slight variations across different SCs \cite{dai2022delay,wu2023parametric}. In contrast, the ultra-large bandwidth in THz band leads to a significantly larger deviation in $\eta_m$ compared to mmWave, making the discrepancy between physical and spatial directions become pronounced and non-negligible.}

In wideband RIS-assisted THz systems, beam split effect occurs at both BS and RIS, resulting in severe performance degradation in channel estimation algorithms \cite{su2023wideband}. Therefore, it is crucial to develop an effective channel estimation scheme to acquire accurate CSI under beam split effect. 

\subsection{Channel Estimation Protocol}
Traditional channel estimation algorithms require substantial modifications to the dictionaries, as conventional designs fail to capture the additional angular information introduced by the beam split effect \cite{elbir2023bsa}. Furthermore, the beam split effect disperses the channel of each SC into different spatial directions, thereby necessitating well-designed channel estimation algorithms that operate independently on each SC, which significantly increases the computational~complexity.

To address these challenges, we propose the following low-complexity channel estimation protocol, which is divided into two main phases, as depicted in Fig. \ref{fig:frame_structure}. Phase I mainly focuses on estimating the full CSI for two selected SCs, which includes estimation of cascaded channels, the DoDs at BS, and the coupled angles at RIS. Specifically, we further divide this phase into two steps. In Step 1, we employ sparse recovery algorithms to estimate the cascaded channel for the selected two SCs. In Step 2, we focus on the heuristic algorithms for common angles estimation. Given the full CSI estimation from Phase I, we then employ a simple LS algorithm to estimate the channels for the remaining SCs in Phase II.

\section{Cascaded Channel Estimation for two SCs} \label{sec:full_CSI_estimation}
In this section, we focus on the estimation of the cascaded channel for two specific SCs. We begin by presenting the angular-domain sparse representation of RIS-assisted THz channels. Next, we introduce the simplified dictionary, referred to as the CBS dictionary. Based on this simplified dictionary, we reformulate the channel estimation problem as a sparse recovery problem. Finally, we provide a comprehensive analysis of existing sparse recovery algorithms.
\subsection{Angular Domain Sparse Representation} 
The sparse angular domain representation of the cascaded channel $\mathbf{H}_{k}^{\mathrm{cas}}[m]$ is given by  
\begin{equation}\label{equ:Dic_H}
\mathbf{H}_{k}^{\mathrm{cas}}[m]=\tilde{\mathbf{\Xi}}_{\mathrm{R},m}\bar{\mathbf{X}}_{k,m}\mathbf{C}_{\mathrm{T},m}^{\mathsf{H}},
\end{equation}
where $\tilde{\mathbf{\Xi}}_{\mathrm{R},m}\triangleq \mathbf{C}_{\mathrm{R},m}^{*}\bullet \mathbf{C}_{\mathrm{R},m} \in \mathbb{C}^{N_{\mathrm{R}} \times Q_{\mathrm{R}}^2}$ is the coupled dictionary. Moreover, $\bar{\mathbf{X}}_{k,m}\triangleq \bar{\boldsymbol{\beta}}_{k,m}\otimes \bar{\mathbf{\Sigma}}_m\in \mathbb{C}^{Q_{\mathrm{R}}^2\times Q_{\mathrm{T}}}
$, where $\bar{\mathbf{\Sigma}}_m\in \mathbb{C}^{Q_{\mathrm{R}}\times Q_{\mathrm{T}}}$ and $\bar{\boldsymbol{\beta}}_{k,m}\in \mathbb{C}^{Q_{\mathrm{R}}}$ are the sparse formulation with $L$ and $J_k$ non-zero entries corresponding to the angle-excluded matrix of BS-RIS channel and the angle-excluded vector of RIS-UE channel, respectively. Consequently, $\bar{\mathbf{X}}_{k,m}$ is an $L\times J_k$-sparse matrix with its non-zero entries corresponding to the set $\{\!\bar{x}_{k\!,m\!,l\!,j}|\bar{x}_{k\!,m\!,l\!,j}\!\triangleq\! \alpha _l\beta _{k,\!j}e^{\!-\jmath 2\pi \!\left( \!\tau _l\!+\!\mu _{k,j} \!\right) \!f_m}\!,\!l\!=\!1,\!\dots \!,\!L,\!j\!\!=\!1,\!\dots\! ,\!J_k \!\} $. Additionally, the overcomplete dictionaries of the BS and RIS directions can be expressed as
\begin{subequations}\label{equ:Codebook}
	\begin{equation}
		\mathbf{C}_{\mathrm{T},m}=\left[ \mathbf{a}_{\mathrm{T},m}\left( \bar{\vartheta}_1 \right) ,\cdots ,\mathbf{a}_{\mathrm{T},m}\left( \bar{\vartheta}_{Q_{\mathrm{T}}} \right) \right],
	\end{equation}
	\begin{equation}
		\mathbf{C}_{\mathrm{R},m}=\left[ \mathbf{a}_{\mathrm{R},m}\left( \bar{\vartheta}_1 \right) ,\cdots ,\mathbf{a}_{\mathrm{R},m}\left( \bar{\vartheta}_{Q_{\mathrm{R}}} \right) \right],
	\end{equation}
\end{subequations}
where $Q_{\mathrm{T}} \gg  N_{\mathrm{T}}$ and $Q_{\mathrm{R}} \gg N_{\mathrm{R}}$ are the quantization level of the BS and RIS directions, respectively. Moreover, $\bar{\vartheta}_q=\delta \cdot q-(Q+1)/{Q} $ with $Q\in\{Q_{\mathrm{T}},Q_{\mathrm{R}}\}$ denotes the spatial direction sample, $\delta = 2/Q$ is the interval between two adjacent samples. 

\subsection{Simplified Dictionary}
Note that the coupled dictionary $\tilde{\mathbf{\Xi}}_{\mathrm{R},m}$ contains only $Q_\mathrm{R}$ unique columns which are exactly the first $Q_\mathrm{R}$ columns in narrowband systems\cite{wang2020compressed}. However, for wideband THz channel with beam split effect, this simplification is no longer suitable due to the frequency-dependent nature of the ARVs.

\begin{proposition}
The dictionary $\tilde{\mathbf{\Xi}}_{\mathrm{R},m}$ contains $Q_B = 2Q_R - 1$ unique columns, and the coupled dictionary can be simplified by including only the first and the last $Q_\mathrm{R}$ columns of the coupled dictionary $\tilde{\mathbf{\Xi}}_{\mathrm{R},m}$, which can be expressed as
	\begin{equation}\label{equ:Xi_simp}
		\mathbf{\Xi }_{\mathrm{R},m}\triangleq [ \tilde{\mathbf{\Xi}}_{\mathrm{R},m}\left( Q_{\mathrm{R}}^{2}-Q_{\mathrm{R}}+1:Q_{\mathrm{R}}^{2} \right) ;\tilde{\mathbf{\Xi}}_{\mathrm{R},m}\left( 2:Q_{\mathrm{R}} \right) ],
	\end{equation}
	where $\mathbf{\Xi}_{\mathrm{R},m}\in\mathbb{C}^{N_{\mathrm{R}}\times Q_\mathrm{B}}$ is the simplified dictionary, and the $q$-th column where $q = 1,\dots Q_\mathrm{B}$ of $\mathbf{\Xi}_{\mathrm{R},m}$ can be formulated~as 
	\begin{equation}\label{equ:xi_q}
		\boldsymbol{\xi} _{\mathrm{R},m,q}=\mathbf{a}_{\mathrm{R},m}\left( 2\frac{q-Q_{\mathrm{R}}}{Q_{\mathrm{R}}} \right).
	\end{equation}
\end{proposition}

\begin{proof}
    The proof is provided in Appendix~\ref{appendix:proof_1}.
\end{proof}


\subsection{Sparse Recovery Reformulation} \label{subsec:sparse_reformulate} By substituting the simplified CBS dictionary $\mathbf{\Xi}_{\mathrm{R},m}$ into \eqref{equ:Dic_H}, the RIS-assisted cascaded channel can be rewritten as
\begin{equation}
\mathbf{H}_{k}^{\mathrm{cas}}[m]=\mathbf{\Xi}_{\mathrm{R},m}\mathbf{X}_{k,m}\mathbf{C}_{\mathrm{T},m}^{\mathsf{H}}.
\end{equation}
Recalling \eqref{equ:Xi_simp} that the simplified dictionary $\mathbf{\Xi}_{\mathrm{R},m}$ is constructed by removing duplicate columns from the original dictionary $\tilde{\mathbf{\Xi}}_{\mathrm{R},m}$, the corresponding angle excluded matrix $\mathbf{X}_{k,m}\in\mathbb{C}^{Q_\mathrm{B}\times Q_{\mathrm{T}}}$ is the merged version of $\bar{\mathbf{X}}_{k,m}\in\mathbb{C}^{Q_\mathrm{B}\times Q_{\mathrm{T}}}$. Specifically,
each row of $\mathbf{X}_{k,m}$ is a superposition of row subsets in $\bar{\mathbf{X}}_{k,m}$, i.e., $\mathbf{X}_{k,m}(i,:)=\sum_{j\in \mathcal{P} _i}{\bar{\mathbf{X}}_{k,m}}(j,:)$. $\mathcal{P}_i$ denotes the set of indices associated with the columns in $\tilde{\mathbf{\Xi}}_{\mathrm{R},m}$ that match the $i$-th column of $\mathbf{\Xi}_{\mathrm{R},m}$. In the cascaded channel, the equivalent angles are the combination of DoDs at BS and coupled angles at RIS, which may overlap and result in fewer effective paths. Therefore, there are ${L}_{\mathrm{tot}}\leq L\times J_k$ non-zero entries in $\mathbf{X}_{k,m}$. Employing the Kronecker product property, the vectorized RIS-assisted cascaded channel $\mathbf{h}_{k}^{\mathrm{cas}}[m] \in \mathbb{C}^{N_{\mathrm{T}}N_{\mathrm{R}} \times Q_\mathrm{B}Q_{\mathrm{T}}}$ can be given by
\begin{equation} \label{equ:Dic_vec_h}
		\mathbf{h}_{k}^{\mathrm{cas}}[m]=\mathbf{\Pi }_m\mathbf{x}_{k,m},
\end{equation}
where $\mathbf{x}_{k,m} = \mathbf{X}_{k,m}(:)\in\mathbb{C}^{Q_\mathrm{B}Q_{\mathrm{T}}}$ represents $\mathbf{X}_{k,m}$ in vector form, $\mathbf{\Pi }_m \triangleq \left( \mathbf{C}_{\mathrm{T},m}^{*}\otimes \mathbf{\Xi }_{\mathrm{R},m} \right) \in \mathbb{C}^{N_{\mathrm{T}}N_{\mathrm{R}} \times Q_\mathrm{B}Q_{\mathrm{T}}}$ denotes the total dictionary. 

By substituting \eqref{equ:Dic_vec_h} into \eqref{equ:y} and vectorizing, the sparse system formulation can be expressed as
\begin{equation} \label{equ:Dic_vec_y}
		\mathbf{y}_{k,t}\left[ m \right] =\left( \mathbf{R}_{t}^{\mathsf{T}}\left[ m \right] \otimes \boldsymbol{\theta }_t \right) \mathbf{\Pi }_m\mathbf{x}_{k,m}+\mathbf {n}_{k,t}[m],
\end{equation}
where $\mathbf{y}_{k,t}\left[ m \right]\in\mathbb{C}^{P}$ is the vector form of $\tilde{\mathbf{y}}_{k,t}[m]$. Stacking the measurement vectors collected at $T$ successive subframes $\mathbf{y}_k\left[ m \right] =[ \mathbf{y}_{k,1}^{\mathsf{T}}\left[ m \right] ,\cdots ,\mathbf{y}_{k,T}^{\mathsf{T}}\left[ m \right] ] ^{\mathsf{T}}\in \mathbb{C} ^{PT}$, the received signal can be expressed as
\begin{equation}\label{equ:c_sparse_problem}
\mathbf{y}_k[m]=\mathbf{\Psi }_m\mathbf{x}_{k,m}+\mathbf{n}_k\left[ m \right],
\end{equation}
where $\mathbf{\Psi }_m=\bar{\mathbf{R}}_m\mathbf{\Pi }_m $ and 
\begin{equation}
\bar{\mathbf{R}}_m\triangleq \left[ \begin{array}{c}
	\mathbf{R}_{1}^{\mathsf{T}}[m]\otimes \boldsymbol{\theta }_1\\
	\vdots\\
	\mathbf{R}_{T}^{\mathsf{T}}[m]\otimes \boldsymbol{\theta }_T\\
\end{array} \right].
\end{equation}
Note that the observation matrix $\mathbf{\Psi }_m\in\mathbb{C}^{PT\times Q_\mathrm{B}Q_{\mathrm{T}}}$ is a sparse matrix. 

To facilitate the application of following sparse recovery algorithms and improve numerical stability \cite{bellili2019generalized}, we transform the complex form in \eqref{equ:c_sparse_problem} to the equivalent real form
\begin{equation}\label{equ:c_r_transform}
\left[ \begin{array}{c}
	\!\!\!\Re \{\mathbf{y}_k\left[ m \right] \}\!\!\!\\
	\!\!\!\Im \{\mathbf{y}_k\left[ m \right] \}\!\!\!\\
\end{array} \right] \!\!\!=\!\!\!\left[ \begin{matrix}
	\Re \{\mathbf{\Psi }_m\}&\!\!\!\!-\Im \{\mathbf{\Psi }_m\}\\
	\Im \{\mathbf{\Psi }_m\}&\!\!\!\!\Re \{\mathbf{\Psi }_m\}\\
\end{matrix} \right]\!\! \left[ \begin{array}{c}
	\!\!\!\!\Re \{\mathbf{x}_{k,m}\}\!\!\!\!\\
	\!\!\!\!\Im \{\mathbf{x}_{k,m}\}\!\!\!\!\\
\end{array} \right] \!+\!\left[ \begin{array}{c}
	\!\!\!\Re \{\mathbf{n}_k\left[ m \right] \} \!\!\!\!\!\\
	\!\!\!\Im \{\mathbf{n}_k\left[ m \right] \}\!\!\!\!\!\\
\end{array} \right]\!\!,
\end{equation}
where $\Re \{\cdot\}$ and $\Im \{\cdot\}$ represent the real and imaginary parts, respectively. By systematically assigning symbols to each matrix in \eqref{equ:c_r_transform}, the sparse recovery problem can be reformulated as
\begin{equation}\label{equ:r_sparse_problem}
	\boldsymbol{y}=\mathbf{\Psi }\boldsymbol{x}+\boldsymbol{n}.
\end{equation}
For simplicity, we adopt a simplified subscript notation by omitting the indices for the matrices and vectors. Let $T_{\mathrm{tot}}\triangleq 2 PT$ and $Q_{\mathrm{tot}}\triangleq 2Q_\mathrm{B}Q_{\mathrm{T}}$. Define $\boldsymbol{y} \in \mathbb{C}^{T_{\mathrm{tot}}}$, $\mathbf{\Psi } \in \mathbb{C}^{T_{\mathrm{tot}}\times Q_{\mathrm{tot}}}$ and $\boldsymbol{x} \in \mathbb{C}^{Q_{\mathrm{tot}}}$ as the measurement vector, the observation matrix and the sparse representation vector to be estimated with sparsity ${L}_{\mathrm{tot}}$, respectively. This sparse recovery formulation forms the basis for the subsequent channel estimation process. In what follows, we analyze different sparse recovery algorithms and motivate the selection of GAMP as the preferred approach for the cascaded channel estimation.
\subsection{Sparse Recovery Algorithm Analysis} 

Selecting an appropriate sparse recovery algorithm is critical for channel estimation in RIS-assisted THz systems due to the unique sparsity characteristics and error propagation effects. This section evaluates three widely used algorithms: LS, OMP, and GAMP.

\textbf{a) LS:} The LS algorithm solves the linear regression problem by minimizing the residual sum of squares. However, it lacks sparsity constraints, which prevents it from effectively capturing the sparse structure of the channel. As a result, LS is unsuitable for the initial full CSI estimation phase, where a sparse representation is required.

\textbf{b) OMP:} OMP is a greedy algorithm that reconstructs the signal by iteratively selecting dictionary atoms that best match the current residual. Since it follows a ``hard estimation" approach, it makes definitive decisions at each iteration without retaining probabilistic information. This approach makes OMP highly sensitive to noise and prone to errors in angle estimation. Without a probabilistic model, it cannot compensate for these errors in later processing, which reduces the robustness of the overall estimation framework.

\textbf{c) GAMP:} Unlike OMP, GAMP adopts a probabilistic message-passing framework that maintains a distribution over signal values throughout the estimation process. This ``soft estimation" mechanism allows GAMP to incorporate uncertainty in the estimation process and mitigate the impact of noise. Instead of selecting a single best estimate at each step, GAMP provides a probability distribution for each angle estimate, which improves the reliability of angular information. This characteristic is particularly advantageous in THz channel estimation, where accurate CSI acquisition directly affects overall system performance.

In summary, GAMP provides superior accuracy and robustness, which makes it the preferred choice for initial full CSI estimation\footnote{A detailed investigation of the CBS-based GAMP scheme is provided in our previous work \cite{su2024channel}.}. Although GAMP can estimate the full channel across all subcarriers, performing this estimation for every subcarrier individually or jointly would lead to substantial computational complexity. To balance estimation accuracy and computational efficiency, we apply GAMP to only two selected subcarriers. The estimated angular parameters are then used to reconstruct the CSI for the remaining SCs.


\section{Common Angle Estimation}\label{sec:angle_estimation}
Utilizing the selected GAMP algorithm from Section \ref{sec:full_CSI_estimation}, we can obtain the sparse signal vector $\boldsymbol{x}$, denoted as $\hat{\mathbf{x}}_{k,m}=\left[ x_1,\cdots ,x_{Q_{\mathrm{tot}}} \right] ^{\mathsf{T}}$. By substituting $\hat{\mathbf{x}}_{k,m}$ into \eqref{equ:Dic_vec_h}, the estimated cascaded channel can be expressed as
\begin{equation}\label{equ:h_cas_hat}
	\hat{\mathbf{h}}_{k}^{\mathrm{cas}}[m]=\mathbf{\Pi }_m\hat{\mathbf{x}}_{k,m}.
\end{equation}

This method implies that we could perform the estimation for each SC independently to obtain the channel estimates for all SCs. However, this requires implementing the estimation algorithm $M$ times, which would result in an unacceptably high computational complexity. While the beam split effect complicates the channel estimation process, it is noticed in \eqref{equ:bs_arv} that different SCs share common physical directions, with the only different term $\eta_m$ \cite{wu2023parametric}. This indicates that the physical angles remain constant across different SCs, which we refer to as the common angles. Inspired by this, we perform the cascaded channel estimation on only a small subset of SCs to acquire the common physical angles. We select $\hat{\mathbf{h}}_{k}^{\mathrm{cas}}[m]$ for $m\in \{1,M/2\}$, which will be explained in Section \ref{subsec:RIS_angle}. Given this estimation, we can further estimate the common directions, i.e., the DoDs at the BS and the coupled directions at the RIS.

\subsection{Common DoDs Estimation at BS}\label{subsec:BS_angle}
Given the estimated cascaded channel $\hat{\mathbf{h}}_{k}^{\mathrm{cas}}[m]=\mathbf{\Pi }_m\hat{\mathbf{x}}_{k,m},\bar{m} \in \left\{ 1,M/2 \right\}$, we propose an energy maximum (EnM) algorithm to estimate the common DoDs at BS\footnote{The two SCs at positions 1 and $M/2$ are selected due to the requirement of coupled angle estimation at RIS in Section \ref{subsec:RIS_angle}.}. Notice that while the EnM algorithm can operate on any single SC, leveraging CSI from multiple SCs improves estimation accuracy \cite{ma2021wideband}. Therefore, we accumulate the CSI obtained from cascaded channel estimation in Step 1. It is noticed that the $q$-th ($q=1,\dots,Q_{\mathrm{tot}}$) entry in $\hat{\mathbf{x}}_{k,\bar{m}}$ corresponds to the $q$-th column in $\mathbf{\Pi }_{\bar{m}}(:,q)$. This means that given the accurate estimation of $\hat{\mathbf{x}}$, we only need to determine its non-zero positions corresponding to the columns in the matrix $\mathbf{\Pi }_{\bar{m}}$. Since $\mathbf{\Pi }_{\bar{m}} \triangleq \left( \mathbf{C}_{\mathrm{T},\bar{m}}^{*}\otimes \mathbf{\Xi }_{\mathrm{R},\bar{m}} \right)$ is in one-to-one correspondence with the ARVs, this relationship allows us to directly obtain the DoDs at BS. However, inevitable noise and estimation errors can cause power leakage in the estimated $\hat{\mathbf{x}}_{k,\bar{m}}$, which will lead to more than ${L}_{\mathrm{tot}}$ non-zero entries. To address this, we propose a heuristic method to determine the accurate non-zero positions, referred to as common physical support, and calculate the corresponding physical DoDs at the BS, as summarized in Algorithm \ref{Alg:energy_maximum}.
\begin{algorithm}[t]
	\caption{EnM Based DoDs Estimation at BS}
	\label{Alg:energy_maximum}
	\begin{algorithmic}[1]
		\REQUIRE Angle-excluded matrix $\mathbf{x}_{k,\bar{m}},\bar{m} \in \left\{ 1,M/2 \right\}$, the path number $L$ of BS-RIS channel and the quantization level of overcomplete dictionaries $Q_{\mathrm{R}}$ and $Q_{\mathrm{T}}$.
		\STATE $\boldsymbol{\epsilon }=\left[ \epsilon _1,\cdots ,\epsilon _{Q_{\mathrm{T}}} \right]$ and $\bar{\boldsymbol{\varphi}}=-1+\frac{1+2\boldsymbol{q}}{Q_{\mathrm{T}}},\boldsymbol{q}=1,\dots ,Q_{\mathrm{T}}  $. 
			\STATE $\tilde{\boldsymbol{\epsilon}}_{\bar{m}}=\mathrm{normalize}\left(\left[ \left| x_{k,\bar{m},1} \right|,\cdots ,\left| x_{k,\bar{m},Q_{\mathrm{tot}}} \right| \right] \right)$.
		
		\FOR {$q \in \left\{1,2,\dots, Q_{\mathrm{T}} \right\}$}
			\STATE	$\epsilon _q=\sum_{\bar{m}\in \left\{ 1,\frac{M}{2} \right\}}{\left\| \tilde{\boldsymbol{\epsilon}}_{\bar{m}}\left( \left( \left( q-1 \right) Q_{\mathrm{B}}+1 \right) :qQ_{\mathrm{B}} \right) \right\| _2}$.
		\ENDFOR
		\STATE $\mathcal{I}^{\mathrm{BS}}=\left\{ q|\mathop {\mathrm{maxk}} \limits_{q}\left( \boldsymbol{\epsilon },L \right) \right\} $
		\STATE $ \hat{\boldsymbol{\varphi}} =\bar{\boldsymbol{\varphi}}\left( \mathcal{I}^{\mathrm{BS} }\right) $.
		\ENSURE $\mathcal{I}^{\mathrm{BS}}$ and $\left\{ \hat{\varphi _l} \right\} _{l=1}^{L}$.
	\end{algorithmic}
\end{algorithm}

In the proposed method, we first establish the normalized energy basis vector $\tilde{\boldsymbol{\epsilon}}_{\bar{m}}$, where each entry is the normalized value of the corresponding element in $\hat{\mathbf{x}}_{k,\bar{m}}$ (Step $2$). As observed in Section \ref{subsec:sparse_reformulate}, $\mathbf{\Pi }_{\bar{m}}$ can be regarded as a block matrix of $Q_{\mathrm{T}}$ blocks, each containing $Q_{\mathrm{B}}$ columns. Since each block in $\mathbf{\Pi }_{\bar{m}}$ is aligned with a physical DoD at BS, $\tilde{\boldsymbol{\epsilon}}_{\bar{m}}$ can be partitioned in the same way. By accumulating the $\ell_2$-norm for each block, we form the merged-energy vector $\boldsymbol{\epsilon}$ to establish a one-to-one correspondence with the physical DoDs at BS (Step $3-5$). \change{Then, the indices corresponding to the $L$ largest elements in $\boldsymbol{\epsilon}$ will form the common physical support $\mathcal{I}^{\mathrm{BS}}$ (Step $6$), where the function $\mathrm{maxk}\left(\boldsymbol{\epsilon},L \right) $ returns the $L$ largest elements in vector $\boldsymbol{\epsilon}$ .} Finally, the DoDs $\left\{ \hat{\varphi _l} \right\} _{l=1}^{L}$ can be obtained using the knowledge of $\mathcal{I}^{\mathrm{BS}}$ and the sampling rule of the CBS dictionary (Step $7$). 

Fig. \ref{fig:BS_angle} depicts an example of the estimated energy spectrum employing the proposed method. The cascaded channel is assumed to contain $L=3$ BS-RIS paths with the common physical support $\mathcal{I} ^{\mathrm{BS}}=\left\{3,4,24\right\}$ and SNR~$=20 $~dB. The spectrum reveals that the energy is distributed surrounding the estimated DoDs which closely matches the actual directions and can be easily distinguished. This validates the effectiveness of the proposed method.

\begin{figure}[!t]
	\centering
	\includegraphics[width=0.47\textwidth]{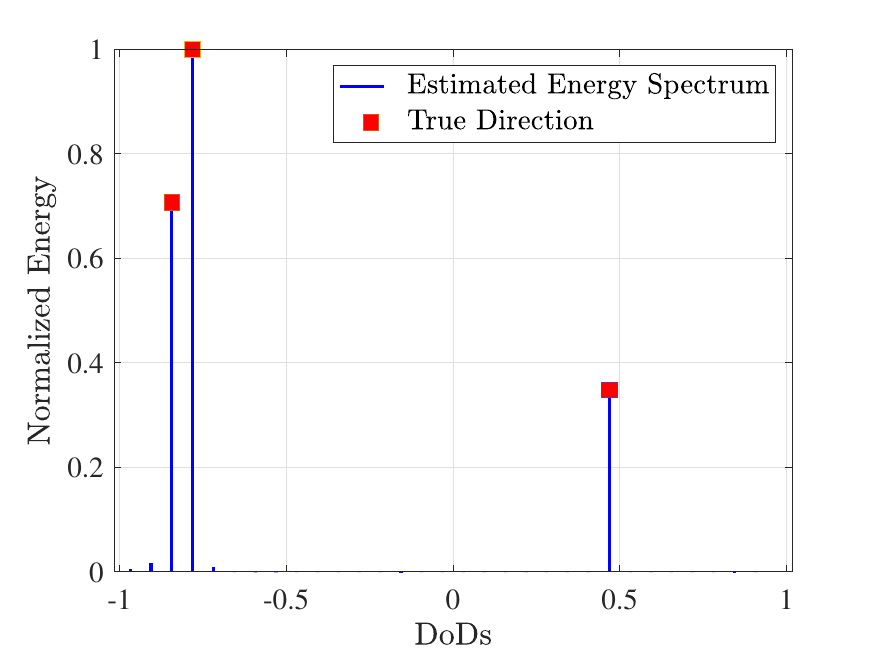}
	\caption{Normalized energy spectrum via EnM algorithm.}
	\label{fig:BS_angle}
\end{figure}
\subsection{Common Coupled Angle Estimation at RIS} \label{subsec:RIS_angle}

With the estimated DoDs $\left\{ \hat{\varphi _l} \right\} _{l=1}^{L}$, the estimated cascaded channel $\hat{\mathbf{H}}_{k}^{\mathrm{cas}}[\bar{m}]$ can be projected onto the subspace of ARM of BS using the pseudo-inverse, which can be given by
\begin{equation}\label{equ:music_u}
	\begin{split}
		\mathbf{U}_{k,\bar{m}}&\triangleq \hat{\mathbf{H}}_{k}^{\mathrm{cas}}[\bar{m}]( \hat{\mathbf{A}}_{N_{\mathrm{T}},\bar{m}}^{\mathsf{H}}\left( \boldsymbol{\varphi } \right) ) ^{\dagger}
		\\
		&=\mathrm{Diag}\left( \mathbf{h}_k[\bar{m}] \right) \mathbf{A}_{\mathrm{R},\bar{m}}\left( \boldsymbol{\psi } \right) \mathbf{\Sigma }_{\bar{m}},
		\end{split}
\end{equation}
where $\mathbf{A}_{N}^{\dagger}\mathbf{A}_N=\mathbf{I}_L$, and $\mathbf{U}_{k,\bar{m}} \triangleq \left[ \mathbf{u}_{k,\bar{m},1},\cdots ,\mathbf{u}_{k,\bar{m},L} \right] \in \mathbb{C} ^{N_{\mathrm{R}}\times L}$ denotes the rest-CSI matrix, which is only related to the DoAs and DoDs at RIS. By substituting \eqref{equ:VAD_G} into \eqref{equ:music_u}, the $l$-th column $\mathbf{u}_{k,\bar{m},l}$ can be expressed as
\begin{equation}\label{equ:music_ul}
	\begin{split}
	\mathbf{u}_{k,\bar{m},l}&=\Sigma _{\bar{m},l}\left( \mathbf{A}_{\mathrm{R},\bar{m}}^{*}\left( \boldsymbol{\vartheta }_k \right) \boldsymbol{\beta }_{k,\bar{m}} \right) \bullet \mathbf{a}_{\mathrm{R},\bar{m}}\left( \psi _l \right) 
	\\
	&=\Sigma _{\bar{m},l}\left[ \mathbf{a}_{\mathrm{R},\bar{m}}\left( \gamma _{k,l,1}\! \right) ,\dots ,\mathbf{a}_{\mathrm{R},\bar{m}}\left( \gamma _{k,l,J_k}\! \right) \right] \!\boldsymbol{\beta }_{k,\bar{m}},
	\end{split}	
\end{equation}
where $\Sigma _{\bar{m},l}$ denotes the $\left(l,l\right)$-th element in $\mathbf{\Sigma }_{\bar{m}}$, and $\gamma _{k,l,j}=\psi _l-\vartheta _{k,j} \in \left[-2,2\right]$ represents the coupled angle for $l=1,\dots ,L$ and $j=1,\dots ,J_k$. 

Define $\mathbf{A}_{\mathrm{R},\bar{m}}\left( \boldsymbol{\gamma }_{k,l} \right) \triangleq\left[ \mathbf{a}_{\mathrm{R},\bar{m}}\left( \gamma _{k,l,1} \right) ,\cdots ,\mathbf{a_{\mathrm{R},\bar{m}}}\left( \gamma _{k,l,J_k} \right) \right] \in \mathbb{C}^{N_{\mathrm{R}} \times J_k} $ as the coupled ARM, \eqref{equ:music_ul} can be rewritten as 
\begin{equation}\label{equ:music_problem}
\mathbf{u}_{k,\bar{m},l}=\mathbf{A}_{\mathrm{R},\bar{m}}\left( \boldsymbol{\gamma }_{k,l,j} \right) \boldsymbol{\chi }_{k,\bar{m},l},
\end{equation}
where $\boldsymbol{\chi }_{k,\bar{m},l}\triangleq \Sigma _{\bar{m},l}\boldsymbol{\beta }_{k,\bar{m}}\in \mathbb{C}^{ J_k}$ denotes the amplitude vector with path gains and phase shifts for each coupled angle. 

The expression in \eqref{equ:music_problem} represents a typical DoA estimation problem, which can be effectively addressed using the MUSIC algorithm due to its high resolution and robustness against noise \cite{kundu1996modified}. However, there are two key challenges in directly applying the MUSIC algorithm. Firstly, the effective angle range for $\boldsymbol{\gamma }_{k,l,j}$, which is the difference between two angles $\psi _l$ and $\theta _{k,j}$, is doubled compared to the classical angle range. Moreover, the relative frequency parameter $\eta_m$ introduced by beam split effect further expands this range. This results in an increased number of potential estimated angles, making it difficult to distinguish the true physical directions. Secondly, the approximate orthogonality of $\boldsymbol{\chi }_{k,\bar{m},l}$ and limited measurements per SC can cause energy leakage into the noise space, further complicating the identification of physical directions. 

To address these challenges, we propose a novel double-sensing MUSIC (DS-MUSIC) algorithm.
Firstly, we combine the information from two selected SCs and apply spatial smoothing preprocessing to the projected rest-CSI matrix $\mathbf{U}_{k,\bar{m}}$. Inspired by the spatial smoothing MUSIC algorithm \cite{liu2015remarks}, we then introduce the proposed DS-MUSIC algorithm to accurately estimate the coupled angles, as described in Algorithm \ref{Alg:music}, where $\mathrm{peaks}\left( \boldsymbol{x},k \right) $ and $\mathrm{mink}\left( \boldsymbol{x},k \right) $ return $k$ largest peaks and $k$ smallest elements of vector $\boldsymbol{x}$, respectively.

We first preprocess the projected signal $\mathbf{u}_{k,\bar{m},l}$ by dividing the RIS array into multiple overlapping subarrays, with each subarray containing $N_\mathrm{sub}$ reflecting elements. There are $N_{\mathrm{a}} = N_{\mathrm{R}} - N_{\mathrm{sub}} + 1$ overlapping subarrays in total. To ensure reliable angle estimation, we enforce the condition $N_{\mathrm{a}} > N_{\mathrm{sub}}$ and $N_{\mathrm{a}}, N_{\mathrm{sub}} > J_k$\footnote{Increasing $N_{\mathrm{a}}$ provides more snapshots for constructing the sample covariance matrix, which improves its estimation accuracy and the stability of the MUSIC algorithm; on the other hand, a larger $N_{\mathrm{sub}}$ improves the angular resolution by increasing the aperture of each subarray. Since $N_{\mathrm{a}} + N_{\mathrm{sub}} - 1 = N_{\mathrm{R}}$ is fixed, these parameters must be jointly chosen to balance spatial resolution and covariance estimation accuracy.}.  The rest-CSI $\mathbf{u}_{k,\bar{m},l,n}$ processed by the $n$-th subarray can be expressed as 
\begin{equation}\label{equ:u_sub}
	\mathbf{u}_{k,\bar{m},l,n}=\mathbf{A}_{\mathrm{R},\bar{m},n}\left( \boldsymbol{\gamma }_{k,l} \right) \boldsymbol{\chi }_{k,\bar{m},l},
\end{equation}
where $\mathbf{u}_{k,\bar{m},l,n}\triangleq \mathbf{u}_{k,\bar{m},l}\left( n:n+N_{\mathrm{sub}}-1 \right) $ and $\mathbf{A}_{\mathrm{R},\bar{m},n}\triangleq \mathbf{A}_{\mathrm{R},\bar{m}}\left( :,n:n+N_{\mathrm{sub}}-1 \right) $. \change{Additionally, $\boldsymbol{\chi }_{k,\bar{m},l}\triangleq \Sigma _{\bar{m},l}\boldsymbol{\beta }_{k,\bar{m}}\in \mathbb{C}^{ J_k}$ denotes the amplitude vector with the path gains and phase shifts for each coupled angle.} Therefore, the covariance matrix $\mathbf{R}_{k,\bar{m},l,n}^{\mathrm{sub}} \in \mathbb{C}^{N_{\mathrm{sub}}\times N_{\mathrm{sub}}}$ can be written as
\begin{equation}\label{equ:R_sub}
	\mathbf{R}_{k,\bar{m},l,n}^{\mathrm{sub}}=\mathbf{u}_{k,\bar{m},l,n}\mathbf{u}_{k,\bar{m},l,n}^{\mathsf{H}}.
\end{equation}

\begin{algorithm}[t]
	\caption{DS-MUSIC Based Coupled DoAs/DoDs Estimation at RIS}
	\label{Alg:music}
	\begin{algorithmic}[1]
		\REQUIRE  Projected matrix $\mathbf{U}_{k,\bar{m}}$, the path number $L$ and $J_k$ of the BS-RIS and RIS-$k$-th UE channel $\forall k \in \mathcal{K}$, the number of subarrays $N_\mathrm{a}$ and quantization level $Q_{\mathrm{R}}$.
		\STATE $[\bar{\boldsymbol{\gamma}}]_q=-2\left( 1+q/Q_{\mathrm{R}} \right) ,q=1,\dots ,Q_{\mathrm{B}}$.
		\FOR {$l \in \left\{1,2,\dots,L\right\}$}
			\FOR {${\bar{m}} \in \left\{1,\frac{M}{2} \right\}$}
				\STATE $\mathbf{R}_{k,\bar{m},l}^{u}=\frac{1}{N_\mathrm{a}}\tilde{\mathbf{U}}_{k,\bar{m},l}\tilde{\mathbf{U}}_{k,\bar{m},l}^{\mathsf{H}}$.
				\STATE $\left[ \tilde{\mathbf{\Omega}}_{k,\bar{m},l},\tilde{\boldsymbol{\lambda}}_{k,\bar{m},l} \right] =\mathrm{eig}\left( \mathbf{R}_{k,\bar{m},l}^{u} \right) $.
				\STATE $\mathcal{S} _l=\left\{ n_{\lambda} | \underset{n_{\lambda}}{\mathrm{mink}}\left( \tilde{\boldsymbol{\lambda}}_{k,\bar{m},l},N_{\mathrm{R}}-J_k \right) \right\} $.
				\STATE $\mathbf{\Omega }_{k,\bar{m},l}^{\mathrm{N}}=\tilde{\mathbf{\Omega }}_{k,\bar{m},l}\left( \mathcal{S} _l \right)  $.
				\FOR {$q \in \left\{1,2,\dots,Q_{\mathrm{B}}\right\}$}
					\STATE $\tilde{\mathbf{p}}_{k,\bar{m},l,q}=\frac{1}{\left\| \left( \mathbf{\Omega }_{k,\bar{m},l}^{\mathrm{N}} \right) ^{\mathsf{H}}\xi _{\mathrm{R},\bar{m},q}^{\mathrm{sub}} \right\|_2 ^2}$.
				\ENDFOR
				\STATE $\mathbf{p}_{k,\bar{m},l}=\left[ \min \left( \tilde{\mathbf{p}}_{k,\bar{m},l} \right) ,\tilde{\mathbf{p}}_{k,\bar{m},l},\min \left( \tilde{\mathbf{p}}_{k,\bar{m},l} \right) \right]$.
				\STATE $\tilde{\mathcal{Q}}_{\bar{m},l}=\left\{ q\left|\underset{q}{\mathrm{maxk}}\left( \mathrm{peaks}\left( \mathbf{p}_{q} \right) ,2J_k \right)\right.\right\} $.
			\ENDFOR
				\STATE $\mathcal{Q} _l=\left\{ \frac{q_1+q_{\frac{M}{2}}^{\prime}}{2}\left| \underset{q_1,q_{\frac{M}{2}}^{\prime}}{\min\mathrm{k}}\left( d\left( q_1,q_{\frac{M}{2}}^{\prime} \right) ,J_k \right)  \right. \right\} $ \\where $q_1\in \tilde{\mathcal{Q}}_{1,l}$ and $q_{\frac{M}{2}}^{\prime}\in \tilde{\mathcal{Q}}_{\frac{M}{2},l}$.
				
		\ENDFOR
		\STATE $\mathcal{I} ^{\mathrm{RIS}}=\mathcal{Q} _1\cup \cdots \cup \mathcal{Q} _L$.
		\STATE $\hat{\boldsymbol{\gamma}}_k=\bar{\boldsymbol{\gamma}}\left( \mathcal{I} ^{\mathrm{RIS}} \right) $.
		\ENSURE  $\mathcal{I} ^{\mathrm{RIS}}$ and $\hat{\boldsymbol{\gamma}}_k$.	
	\end{algorithmic}
\end{algorithm}

To implement the spatial smoothing, the covariance matrix $\mathbf{R}_{k,\bar{m},l}^{u} \in \mathbb{C}^{N_{\mathrm{sub}}\times N_{\mathrm{sub}}}$ of $\mathbf{u}_{k,\bar{m},l}$ is calculated by averaging all $N_\mathrm{a}$ subarray covariance matrices $\mathbf{R}_{k,\bar{m},l,n}^{\mathrm{sub}}$, which can be expressed by
\begin{equation}\label{equ:R_u}
		\mathbf{R}_{k,\bar{m},l}^{u}=\frac{1}{N_\mathrm{a}}\sum_{n=1}^{N_\mathrm{a}}{\mathbf{R}_{k,\bar{m},l,n}^{\mathrm{sub}}}=\frac{1}{N_\mathrm{a}}\tilde{\mathbf{U}}_{k,\bar{m},l}\tilde{\mathbf{U}}_{k,\bar{m},l}^{\mathsf{H}},
\end{equation}
where $\tilde{\mathbf{U}}_{k,\bar{m},l} \triangleq \left[ \mathbf{u}_{k,\bar{m},l,1},\dots ,\mathbf{u}_{k,\bar{m},l,N_\mathrm{a}} \right] \in \mathbb{C} ^{N_{\mathrm{sub}}\times N_\mathrm{a}}$ denotes the concatenation of the subspace-supporting vectors $\mathbf{u}_{k,\bar{m},l,n}$ as defined in \eqref{equ:u_sub}. The eigendecomposition of $\mathbf{R}_{k,\bar{m},l}^{u}$ can be expressed as
\begin{equation}\label{equ:eigR_u}
\mathbf{R}_{k,\bar{m},l}^{u}=\tilde{\mathbf{\Omega}}_{k,\bar{m},l}\mathrm{Diag}( \tilde{\boldsymbol{\lambda}}_{k,\bar{m},l} ) \tilde{\mathbf{\Omega}}_{k,\bar{m},l}^{\mathsf{H}},
\end{equation}
where  $\tilde{\boldsymbol{\Omega}}_{k,\bar{m},l}\triangleq\left[ \tilde{\boldsymbol{\omega}}_{k,\bar{m},l,1},\cdots ,\tilde{\boldsymbol{\omega}}_{k,\bar{m},l,N_{\mathrm{sub}}} \right] \in \mathbb{C} ^{N_{\mathrm{sub}}\times N_{\mathrm{sub}}}$ denotes the eigenbasis matrix which contains  $N_{\mathrm{sub}}$ eigenvectors and $\tilde{\boldsymbol{\lambda}}_{k,\bar{m},l} \in \mathbb{C} ^{N_{\mathrm{sub}}}$ is the eigenvalue vector. Note that $\mathbf{R}_{k,\bar{m},l}^{u}$ is a Hermitian matrix, which means that all its $N_{\mathrm{sub}}$ eigenvectors are orthogonal to each other. By arranging the eigenvalues of $\mathbf{R}_{k,\bar{m},l}^{u}$ in descending order and organizing the corresponding eigenvectors accordingly, it can be rewritten as
\begin{equation}\label{equ:eigR_u_sorted}
	\mathbf{R}_{k,\bar{m},l}^{u}=\mathbf{\Omega}_{k,\bar{m},l}\mathrm{Diag}\left( \boldsymbol{\lambda}_{k,\bar{m},l} \right) \mathbf{\Omega}_{k,\bar{m},l}^{\mathsf{H}},
\end{equation}
where $\mathbf{\Omega}_{k,\bar{m},l}$ and $\boldsymbol{\lambda}_{k,\bar{m},l}$ denote the sorted eigenbasis matrix and eigenvalue vector, respectively. Therefore, the matrix $\mathbf{\Omega }_{k,\bar{m},l}\triangleq[ \mathbf{\Omega }_{k,\bar{m},l}^{\mathrm{S}}\,\,\mathbf{\Omega }_{k,\bar{m},l}^{\mathrm{N}} ] $ can be divided into two parts: $\mathbf{\Omega }_{k,\bar{m},l}^{\mathrm{S}}\in \mathbb{C} ^{N_{\mathrm{sub}}\times J_k}$ and $\mathbf{\Omega }_{k,\bar{m},l}^{\mathrm{N}}\in \mathbb{C} ^{N_{\mathrm{sub}}\times \left(N_{\mathrm{sub}}-J_k\right)}$. The part $\mathbf{\Omega }_{k,\bar{m},l}^{\mathrm{S}}$ consists of the $J_k$ largest eigenvalues corresponding to the directions of largest variability, which spans the signal subspace. The $N_{\mathrm{sub}}-J_k$ remaining eigenvectors, which form $\mathbf{\Omega }_{k,\bar{m},l}^{\mathrm{N}}$, span the noise subspace and are orthogonal to the signal subspace.


To calculate the MUSIC pseudo-spectrum, the practical implementation requires sampling the estimated angles. We define $\mathbf{\tilde{p}}_{k,\bar{m},l}\triangleq [ \tilde{p}_{k,\bar{m},l,1},\dots ,\tilde{p}_{k,\bar{m},l,Q_{\mathrm{B}}} ]^{\mathsf{T}} $ as the MUSIC pseudo-spectrum, where $\tilde{p}_{k,\bar{m},l,q}$ is the sampling spectral value for $q=1,\dots,Q_{\mathrm{B}}$. To align these sampled values with our proposed algorithm, we adjust the angle sampling within the range of the CBS dictionary. Therefore, we define the sub-CBS dictionary $\mathbf{\Xi }_{\mathrm{R},\bar{m}}^{\mathrm{sub}} \in \mathbb{C} ^{N_{\mathrm{sub}}\times Q_{\mathrm{B}} }$ as the first $N_{\mathrm{sub}}$ rows of the original CBS dictionary in \eqref{equ:Xi_simp}. According to \eqref{equ:xi_q}, each element in the sub-CBS dictionary $\boldsymbol{\xi} _{\mathrm{R},\bar{m},q}^{\mathrm{sub}}$ can be derived as
\begin{equation}\label{equ:element_sub}
	\boldsymbol{\xi} _{\mathrm{R},\bar{m},q}^{\mathrm{sub}}\triangleq e^{-\jmath2\pi \eta _m\left( \frac{q}{Q_{\mathrm{R}}}-1 \right) \mathbf{n}_{\mathrm{sub}}} ,q=1,\dots ,Q_{\mathrm{B}},
\end{equation}
where $\mathbf{n}_{\mathrm{sub}}=\left[ 1,\cdots ,N_{\mathrm{sub}} \right] $.

By mapping the sampled beams $\boldsymbol{\xi} _{\mathrm{R},\bar{m},q}^{\mathrm{sub}}$ onto the noise space $\mathbf{\Omega }_{k,\bar{m},l}^{\mathrm{N}}$, the MUSIC algorithm constructs the pseudo-spectrum by evaluating the inverse of the projection's norm at each potential direction, which can be given as
\begin{equation}\label{equ:P_music}
	\tilde{p}_{k,\bar{m},l,q}=\frac{1}{\left\| ( \mathbf{\Omega }_{k,\bar{m},l}^{\mathrm{N}} ) ^{\mathsf{H}}\boldsymbol{\xi} _{\mathrm{R},\bar{m},q}^{\mathrm{sub}} \right\|_2 ^2},q=1,\dots ,Q_{\mathrm{B}}.
\end{equation}
Therefore, the MUSIC pseudo-spectrum can be express as
\begin{equation}\label{equ:P_music_ext}
\mathbf{p}_{k,\bar{m},l}=\left[ \tilde{p}_{k,\bar{m},l,1},\cdots ,\tilde{p}_{k,\bar{m},l,Q_{\mathrm{B}}} \right].
\end{equation}

Given that any beam $\boldsymbol{\xi} _{\mathrm{R},\bar{m},q}^{\mathrm{sub}}\in \mathbf{\Omega }_{k,\bar{m},l}^{\mathrm{S}}$ in the signal subspace should ideally have minimal projection onto the noise subspace $ \mathbf{\Omega }_{k,\bar{m},l}^{\mathrm{N}}$, which means that the $\ell_2$-norm in \eqref{equ:P_music} associated with the actual directions should approximate zero. The MUSIC algorithm typically relies on peak detection to find actual physical directions. 
The potential directions can then be determined by employing local-maximum algorithms. However, due to the extended range of coupled angles, the estimated number of directions will be $K_{m,\text{cp}}$ times the original with $K_{\text{cp}}<K_{m,\text{cp,max}}$.
\begin{proposition}
	\label{prop:max_k}
    The maximum number of identical beams can be given as
    \begin{equation}\label{equ:num_ident_beam}
		K_{m,\rm{cp},\max}=\lfloor \eta _m\left( 2-1/Q_{\mathrm{R}} \right) \rfloor + 1,
	\end{equation}
	\change{where $\lfloor\cdot\rfloor$ is the floor function.}
\end{proposition}
\begin{proof} 
Recalling from \eqref{equ:element_sub} that the sampled beams $\boldsymbol{\xi}_{\mathrm{R},m,q}^{\mathrm{sub}}$ can be written as $\boldsymbol{\xi} _{\mathrm{R},m,q}^{\mathrm{sub}}=e^{-\jmath2\pi \eta _m\left( \frac{q}{Q_{\mathrm{R}}}-1 \right)\mathbf{n}_{\mathrm{sub}}}$. Based on the periodicity property of the complex exponential function, which satisfies $e^{-\jmath2\pi K_{m,\text{cp}}}  =1$ for any integer $K_{m,\text{cp}} \in \mathbb{N}$, $\boldsymbol{\xi}_{\mathrm{R},m,q}^{\mathrm{sub}}$ can be equivalently expressed as
\begin{equation}\label{equ:xi_period}
		\begin{split}
			\boldsymbol{\xi} _{\mathrm{R},m,q}^{\mathrm{sub}}
			&=e^{-\jmath2\pi K_{m,\text{cp}}\mathbf{n}_{\mathrm{sub}}}e^{-\jmath2\pi \eta _m\left( \frac{q}{Q_{\mathrm{R}}}-1 \right) \mathbf{n}_{\mathrm{sub}}}
			\\
			&=e^{-\jmath2\pi \eta _m\left( \frac{q}{Q_{\mathrm{R}}}-1+\frac{K_{m,\text{cp}}}{\eta _m} \right) \mathbf{n}_{\mathrm{sub}}}.
		\end{split}
	\end{equation}
It is observed that when $q^{\prime} = q+K_{m,\text{cp}}Q_{\mathrm{R}}/\eta _m $, the element in sub-CBS dictionary will generate the identical beam, i.e., $\boldsymbol{\xi} _{\mathrm{R},m,q}^{\mathrm{sub}}=\boldsymbol{\xi} _{\mathrm{R},m,q'}^{\mathrm{sub}}$. Therefore, the maximum number of identical beams $K_{m,\text{cp},\max}$ can be obtained by solving the integer equation $K_{m,\text{cp}}Q_{\mathrm{R}}/\eta _m=Q_{\mathrm{B}}$. By rearranging the terms and employing the relationship $Q_{\mathrm{B}}=2Q_{\mathrm{R}}-1$, we can obtain
	\begin{equation}
		K_{m,\text{cp}}=\frac{Q_{\mathrm{B}}\eta _m}{Q_{\mathrm{R}}}.
	\end{equation}
	
	Therefore, the maximum value of $K_{m,\text{cp}}$ is the integer part of this expression, which can be given as
	\begin{equation}\label{equ:K_cp_max_floor}
		K_{m,\rm{cp},\max} = \left\lfloor \eta_m \left( 2 - \frac{1}{Q_{\mathrm{R}}} \right) \right\rfloor + 1.
	\end{equation}
\end{proof}

According to Proposition \ref{prop:max_k}, instead of the expected $J_k$ peaks in the MUSIC pseudo-spectrum, $K_{\text{cp},\max} J_k$ peaks will be detected. This makes the algorithm infeasible for distinguishing the true physical directions from fake ones within a single SC. 
Thankfully, it is observed from \eqref{equ:xi_period} that the period $K_{m,\text{cp}}/\eta_m$ varies with SC frequency. Therefore, the periodic false spatial directions are frequency-dependent, while the true physical directions remain frequency-independent \cite{wu2023parametric}. Inspired by this, we propose a novel double sensing (DS) algorithm to acquire the true physical direction. In this approach, we select two SCs to calculate and compare their MUSIC pseudo-spectrum, as these two SCs provide the necessary information to uniquely determine the coupled angles at RIS.

Assuming $\theta_{q_{\mathrm{p}}}$ to be the true physical direction, the corresponding sampled-beam $\boldsymbol{\xi}_{\mathrm{R},\bar{m},q_{\mathrm{p}}}$ can be determined according to \eqref{equ:P_music}.
Then, the unique angular index $q_{\mathrm{p}}$ can be easily decoupled from \eqref{equ:element_sub}, which is frequency-independent and remains consistent across all the SCs\cite{wu2023parametric}. In contrast, the index of periodic false angles, given by $q^{\prime}_{\mathrm{p}} = q_{\mathrm{p}}+K_{m,\text{cp}}Q_{\mathrm{R}}/\eta _m$, varies among different SCs. Therefore, the positions where peaks overlap indicate the true physical directions. However, due to unavoidable noise and estimation errors, the pseudo-spectrum across different SCs may exhibit offsets. To estimate the true physical directions accurately, we first identify the $J_k$ closest peaks, which correspond to the overlapping peaks across different SCs. Therefore, we introduce the peak distance $d$ to quantify this deviation, given by:
\begin{equation}\label{equ:dist}
d\left( q_m,q_{m^{\prime}}^{\prime} \right) =\left| q_m-q_{m^{\prime}}^{\prime} \right|,
\end{equation}
where $q_m$ and $q_{m^{\prime}}^{\prime}$ denote the indices of the peaks at two selected SCs, respectively. So far, the detailed steps of the proposed DS-MUSIC algorithm are summarized in Algorithm~\ref{Alg:music}.

Selecting the closest peaks ensures consistency in the estimated directions and mitigates the impact of frequency-dependent variations. However, while the closest peaks are used for estimating the true directions, it is also essential to ensure that 1) the number of non-overlapping peaks are limited, and 2) these non-overlapping peaks are sufficiently separated. Satisfying these conditions allows effective discrimination between true and false physical directions. This can be achieved by selecting the appropriate SCs.
Specifically, since $K_{m,\text{cp},\max}$ in \eqref{equ:num_ident_beam} is affected by the choice of SCs, we should aim to satisfy $K_{m,\text{cp},\max}\leq2$, which means each angle has at most two identical peaks. \change{Employing the floor function inequality $\lfloor x \rfloor \le x$ and substituting \eqref{equ:f_m} into \eqref{equ:K_cp_max_floor}, we can obtain}
\begin{equation}\label{equ:m}
	m<\frac{M}{2}+\frac{Mf_c}{B Q_\mathrm{B}}+\frac{1}{2}.
\end{equation}
Besides, substituting \eqref{equ:f_m} into \eqref{equ:dist}, the peak distance $d$ can be derived as
\begin{equation}\label{equ:dist2}
	d\left( q_m,q_{m^{\prime}}^{\prime} \right) =K_{m,\text{cp}}Q_{\mathrm{R}}f_c\left| \frac{1}{f_{m^{\prime}}}-\frac{1}{f_m} \right|.
\end{equation}

To distinguish true physical directions effectively, the distance among non-overlapping peaks should also be sufficiently large. According to \eqref{equ:dist2}, it is evident that the intervals between selected SCs should be large enough. Therefore, under the constraint in \eqref{equ:m}, we select the first and the $M/2$-th SCs to enhance the accuracy of estimation\footnote{Note that while using more subcarriers may offer additional spatial information, our two-subcarrier design leverages the deterministic beam split property for efficient peak alignment. Extending the alignment strategy to more than two subcarriers would require complex correspondence matching, which may introduce ambiguity and computational overhead without significant accuracy improvement.}.

Fig. \ref{fig:RIS_angle} depicts an example of estimated MUSIC pseudo-spectrum utilizing the proposed DS approach. Assuming that the cascaded channel contains $J_k=3$ RIS-UE paths for an arbitrary UE with SNR~$=20 $~dB. It can be observed that the peaks overlap at actual physical directions, while those associated with fake spatial directions are completely separated. The estimated physical directions can be easily distinguished from the fake ones and accurately match the actual physical directions.

\begin{figure}[!t]
	\centering
	\includegraphics[width=0.47\textwidth]{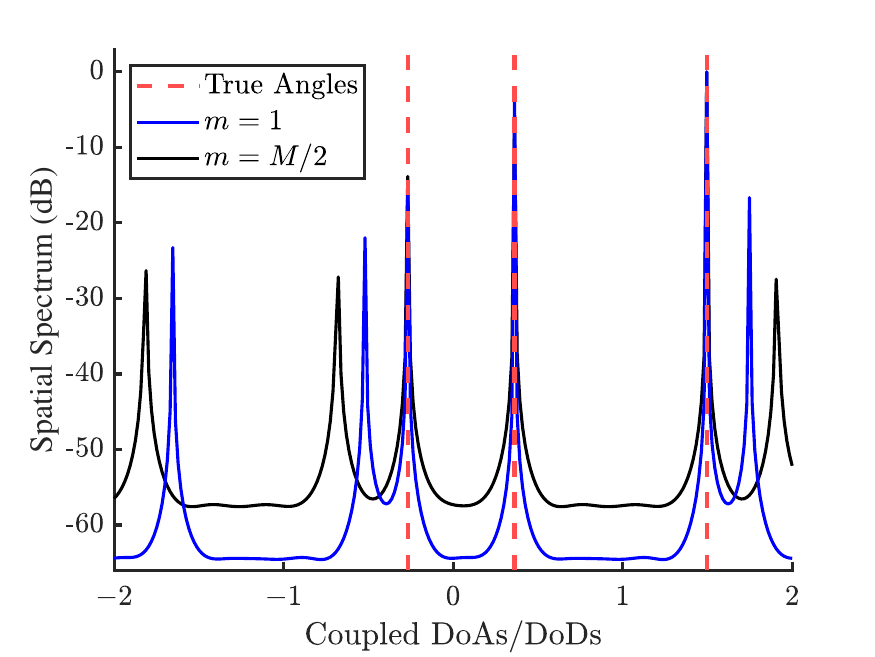}
	\caption{\change{Normalized MUSIC pseudo-spectrum via DS-MUSIC algorithm.}}
	\label{fig:RIS_angle}
\end{figure}

\section{Channel Estimation for Other SCs}\label{sec:CE_other_SC}
Through the proposed two-step full CSI estimation, we obtained the cascaded channels for two SCs and the corresponding common physical angle. Based on this knowledge, the cascaded channel in \eqref{equ:VAD_H} can be rewritten as
\begin{equation}\label{equ:LS_h} 
	\mathbf{h}_{k}^{\mathrm{cas}}[m]=\left( \mathbf{A}_{\mathrm{T},m}^{*}\left( \hat{\boldsymbol{\varphi}} \right) \otimes \mathbf{A}_{\mathrm{R},m}\left(\hat{\boldsymbol{\gamma}}_k \right) \right) \tilde{\mathbf{x}}_{k,m},
\end{equation}
where $\tilde{\mathbf{x}}_{k,m}\!\triangleq \!\boldsymbol{\beta }_{k,m}\!\otimes\! \mathbf{\Sigma }_m\!\in \mathbb{C} ^{L^2 \!J_k}$ represents the angle-excluded vector, $\hat{\boldsymbol{\varphi}}\triangleq [ \hat{\varphi}_1,\dots ,\hat{\varphi}_L ] \in \mathbb{C} ^L$ is the estimated DoDs at BS of all $L$ paths, and $\hat{\boldsymbol{\gamma}}_k\triangleq [ \hat{\boldsymbol{\gamma}}_{k,1}^{\mathsf{T}},\dots ,\hat{\boldsymbol{\gamma}}_{k,L}^{\mathsf{T}} ] ^{\mathsf{T}}\in \mathbb{C} ^{LJ_k}$ denotes the estimated coupled directions at RIS. 

Substituting \eqref{equ:LS_h} into \eqref{equ:Dic_vec_y}, the received signal $\mathbf{y}_{k,t}\left[ m \right]$ can be rewritten as
\begin{equation}\label{equ: y_ls}
		\mathbf{y}_{k,t}\left[ m \right]=\tilde{\mathbf{\Psi}}_m\tilde{\mathbf{x}}_{k,m},
\end{equation}
where $\tilde{\mathbf{\Psi}}_m\!\triangleq\bar{\mathbf{R}}_m\left( \mathbf{A}_{\mathrm{T},m}^{*}\left( \hat{\boldsymbol{\varphi}} \right) \!\otimes\! \mathbf{A}_{\mathrm{R},m}\left( \hat{\boldsymbol{\gamma}}_k \right) \right)\! \in \!\mathbb{C} ^{PT\times L^2J_k}$. The column size of the $\tilde{\mathbf{\Psi}}_m$ matrix is reduced from $Q_{\mathrm{B}}Q_{\mathrm{T}}$ of the overcomplete dictionary to a multipath-related size $L^2J_k$. Moreover, the incorporation of angular information transforms the observation matrix $\tilde{\mathbf{\Psi}}_m$ transformed into a full-rank matrix. Consequently, the sparse representation in \eqref{equ:c_sparse_problem} is reverted to the non-sparse frequency domain. This significant reduction in the column size and reconstruction of signal structure reformulate the channel estimation problem of remaining SCs a low-dimensional linear regression problem, enhancing both the efficiency and feasibility of the LS solution.  
Therefore, by employing LS algorithm, the angle-excluded coefficients $\tilde{\mathbf{x}}_{k,m}$ can be easily obtained as
\begin{equation}
\label{equ:exclude_x}
\hat{\tilde{\mathbf{x}}}_{k,m}=\left( \tilde{\mathbf{\Psi}}_{m}^{\mathsf{H}}\tilde{\mathbf{\Psi}}_m \right) ^{-\mathsf{1}}\tilde{\mathbf{\Psi}}_{m}^{\mathsf{H}}\mathbf{y}_{k,t}\left[ m \right].
\end{equation}

Substituting \eqref{equ:exclude_x} into \eqref{equ:LS_h}, the full CSI for other SCs $m \in \mathcal{M}\backslash\{1,M/2\}$ can be expressed as
\begin{equation}
\hat{\mathbf{h}}_{k}^{\mathrm{cas}}[m]=\left( \mathbf{A}_{\mathrm{T},m}^{*}\left( \hat{\boldsymbol{\varphi}} \right) \otimes \mathbf{A}_{\mathrm{R},m}\left( \hat{\boldsymbol{\gamma}}_k \right) \right) \hat{\tilde{\mathbf{x}}}_{k,m}.
\end{equation}

\section{Analysis of Computational Complexity}\label{sec:complexity_analysis}
In this section, we analyze the computational complexity of the proposed channel estimation protocol and compare it with that of existing algorithms. Specifically, the overall complexity of the proposed scheme is primarily attributed to the following four components:
\begin{enumerate}
	\item Cascaded channel estimation for two selected SCs by CBS-based GAMP approach: $\mathcal{O} \left( 2PTQ_{\mathrm{T}}Q_\mathrm{B}N_\mathrm{iter}\right) $, where $N_\mathrm{iter}$ denotes the number of iterations required for GAMP to converge.
	\item DoDs acquisition at BS: $\mathcal{O} \left( Q_{\mathrm{T}}Q_\mathrm{B} \right) $.
	\item Coupled directions acquisition at the RIS using DS-MUSIC algorithm: $\mathcal{O} \left( 2N_{\mathrm{R}}^{3} \right) $.
	\item Channel estimation for other SCs: $\mathcal{O} \left( \left( M-2\right)L^6J_{k}^{3} \right) $.
\end{enumerate}

The overall complexity of the proposed scheme can be expressed as
\begin{equation}
\mathcal{O} \left( 2PTQ_{\mathrm{T}}Q_{\mathrm{B}}N_\mathrm{iter}+2N_{\mathrm{R}}^{3}+\left( M-2 \right) L^6J_{k}^{3} \right). 
\end{equation}

For comparison, we also discuss the complexity of the conventional OMP scheme \cite{tropp2007signal}, BSA-OMP scheme \cite{elbir2023bsa} and CBS-GAMP scheme \cite{su2024generalized}. These algorithms perform under traditional channel estimation frameworks which means they need to be executed at each SC or simultaneously process all SCs. 
The computational complexity of conventional OMP algorithm is $\mathcal{O} \left( MPTQ_{\mathrm{T}}{Q_{\mathrm{R}}^2} \right)  $. The BSA-OMP approach has the same complexity as the conventional OMP, $\mathcal{O} \left( MPTQ_{\mathrm{T}}{Q_{\mathrm{R}}^2} \right) $, since it similarly processes each SC independently. The computational complexity of CBS-GAMP is $\mathcal{O} \left(MPTQ_{\mathrm{T}}Q_\mathrm{B}N_\mathrm{iter} \right) $. Additionally, the number of SCs $M$ is typically large in wideband systems. Therefore, it is evident that the computational complexity of the proposed scheme is significantly lower than the other approaches mentioned.


\section{Numerical Results and Discussion}\label{sec:Numerical Results and Discussion}
In this section, we present numerical results compared with both mmWave and THz benchmarks to validate the effectiveness of the proposed scheme. Without loss of generality, we consider wideband RIS-assisted communication scenarios with $M=128$ SCs. A BS employing a ULA with $N_{\mathrm{T}}$ = 16 antennas transmits signals to $K=8$ single antenna UEs, with the assistance of a RIS equipped with $N_{\mathrm{R}}= 64$ passive reflecting elements. Moreover, the hybrid beamformer vector $\mathbf{w}_{t,p}[m]$ is assumed to be randomly distributed \cite{elbir2023bsa}.

\subsection{Comparison with mmWave Solution}
\label{subsec:mmwave}
In this subsection, we present simulation results to illustrate the distinct characteristics and impacts of beam squint and beam split effects, and further demonstrate the unique advantages of the proposed framework in comparison with mmWave~solutions.
\subsubsection{Setup}
For the mmWave system, the central carrier frequency is set to $f_c= 28$ GHz with a bandwidth of $B = 600$ MHz, following the configuration in \cite{liu2021cascaded}. For the considered THz system, the central carrier frequency is increased to $f_c= 100$  GHz with a wider bandwidth of $B = 15$ GHz. Furthermore, the channel characterization adopts the same parameters as those specified in \cite{tan2021wideband}, and the path gains are model as $\alpha_l \sim \mathcal{CN}(0,1)$ and the time delays are constrained by $\tau_l \leq \tau_{\max}$ with $\tau^{\mathrm{mmWave}}_{\max}=53$ ns and $\tau^{\mathrm{THz}}_{\max}=20$ ns for mmWave and THz systems, respectively.

Additionally, we assume that the DoDs at BS are available as prior knowledge in this subsection, since the considered benchmark focused only on estimating the angles at RIS. Accordingly, the number of propagation paths is set to $L=1$ and $J_k=2, \forall k$, respectively. Furthermore, the DoAs and DoDs are randomly generated within the range $\left[-\pi/2,\pi/2\right]$. To further highlight the effect of beam squint and beam split, we consider the DoDs at RIS to be under large-angle reflection conditions, which significantly deviate from the normal incidence.
\subsubsection{Benchmark}
We adopt the NOMP based mmWave solution proposed in \cite{liu2021cascaded} as the benchmrk method. This approach is implemented for both mmWave and THz band, referred to as `NOMP-mmWave' and `NOMP-THz', respectively.

\subsubsection{Metric}
The root mean square error (RMSE) is utilized to evaluate the accuracy of angle estimation, defined~as
\begin{equation}
\mathrm{RMSE}=\sqrt{\mathbb{E} \left\{ \left| \theta -\hat{\theta} \right| \right\}},
\end{equation}
where $\theta$ and $\hat{\theta}$ represent the true angle and the estimated angle, respectively.

\subsubsection{Performance Analysis}
For a fair comparison, identical pilot resources are allocated to both the NOMP-based solution and the proposed approach for estimating cascaded angles at RIS. Let $M_{\mathrm{sub}}$ denote the number of SCs used for angle estimation. The resource parameters of NOMP algorithm are set to $T = 50 $ and $ M_{\mathrm{sub}}=12$, with the hybrid beamforming functionality disabled, consistent with the original implementation settings in \cite{liu2021cascaded}. Accordingly, the proposed method utilizes $T=25$, $P=12$ and $M_{\mathrm{sub}} = 2$ resource~elements.

Fig.~\ref{fig:RMSE_SNR} shows the RMSE performance of the coupled DoAs/DoDs estimation at RIS. While the NOMP-based mmWave solution achieves satisfactory performance in conventional mmWave scenarios, it fails to adequately address the beam split effect in THz systems. \change{This degradation is primarily attributed to the substantially larger SC spacing in THz band, which induces more severe beam misalignment compared to the relatively mild beam squint effect observed at mmWave frequencies. As a result, conventional mmWave solutions fail to mitigate this misalignment by merely accumulating information from multiple SCs. In contrast, the proposed approach integrates a specifically designed DS-MUSIC algorithm, which distinguishes the true angular components from the frequency-dependent false ones by leveraging the inter-SC misalignment within an extended angular range, thereby achieving significantly enhanced estimation accuracy.} Moreover, increasing the angular dictionary size $Q_{\mathrm{R}}$ leads to additional performance gains by mitigating quantization errors in the angle domain.
\begin{figure}[!t]
	\centering
	\includegraphics[width=0.48\textwidth]{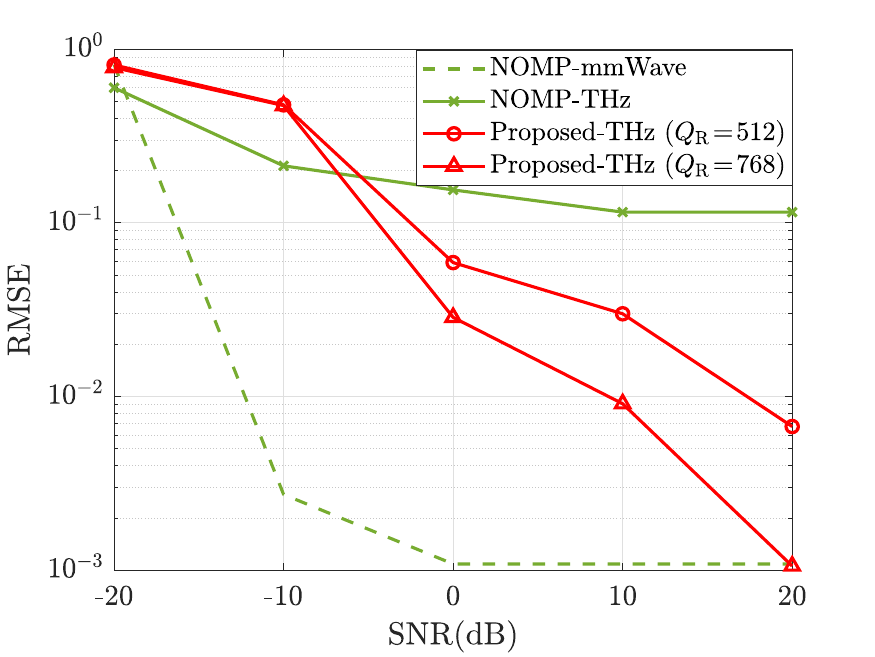}
	\caption{RMSE performance of coupled DoAs/DoDs at RIS against SNR.}
	\label{fig:RMSE_SNR}
\end{figure}
\subsection{Comparison with THz Solutions}
In this subsection, we investigate the impact of different system parameters on performance, and provide a comprehensive comparison with state-of-the-art algorithms in THz band.
\subsubsection{Setup}
We consider the THz system as defined in Section \ref{subsec:mmwave}. Without loss of generality, the DoAs and DoDs are randomly generated from the discretized grid within the range $\left[-\pi/2,\pi/2\right]$ in this subsection \cite{dovelos2021channel} \footnote{The quantization errors introduced by discrete angular grids assumption are small for large dictionary quantization level $Q_{\mathrm{T}}$ and $Q_{\mathrm{R}}$, which do not affect significantly the normalized array gain and can be neglected.}. Consequently, the quantization level of overcomplete dictionaries can be reduced to $Q_{\mathrm{T}}= 32$ and $Q_{\mathrm{R}}= 128$ for simplicity, and the 1-st and 64-th SCs are selected for the full CSI estimation in Phase I Step 1. Furthermore, we consider a more practical multipath model where the location of both BS and RIS is unknown and there exists $L=J_k=3, \forall k$ propagation paths for both BS-RIS and RIS-UEs channels.

\subsubsection{Benchmarks}
For comparison, we also evaluate the following benchmark schemes in the simulation. Unless otherwise specified, they are configured identically to the proposed method to ensure a fair comparison.
\begin{itemize}
	\item \textit{Conventional-OMP} \cite{tropp2007signal}: Traditional dictionaries are adopted to span the beam space and the CSI is reconstructed by OMP algorithm at each SC.
	\item \textit{BSA-OMP} \cite{elbir2023bsa}: Beam-split-adapted dictionaries are employed to cover the beam space and the CSI is recover by OMP algorithm at each SC.
	\item \textit{CBS-GAMP} \cite{su2024generalized}: CBS dictionaries are employed to cover the beam space and the CSI is recover by GAMP algorithm at each SC. 
	\item \textit{Oracle-LS}: The angular information is assumed perfectly known and the angle-excluded information is estimated by LS algorithm. This algorithm is regarded as the performance upper bound of the proposed scheme.
\end{itemize}

\subsubsection{Metrics}
We use the normalized mean square error (NMSE) as a metric for full CSI estimation performance, which is given by 
\begin{equation}
\mathrm{NMSE}=\mathbb{E} \left\{ \frac{\left\| \mathbf{h}_{k}^{\mathrm{cas}}[m]-\mathbf{\hat{h}}_{k}^{\mathrm{cas}}[m] \right\| _{2}^{2}}{\left\| \mathbf{h}_{k}^{\mathrm{cas}}[m] \right\| _{2}^{2}} \right\}.
\end{equation}

Additionally, we employ the correct probability as a metric for direction estimation performance, which is given by 
\begin{equation}
\mathrm{P}_{\mathrm{c}}=\mathbb{E} \left\{ \mathrm{Pr}\left[ \left| \theta -\hat{\theta} \right|<\varepsilon _{\mathrm{th}} \right] \right\} ,
\end{equation}
where $\varepsilon _{\mathrm{th}}\!=\!\delta/2$ denotes the threshold which is set as half the sampling interval.

\subsubsection{Performance Analysis}
\begin{figure}[!t]
	\centering
	\includegraphics[width=0.48\textwidth]{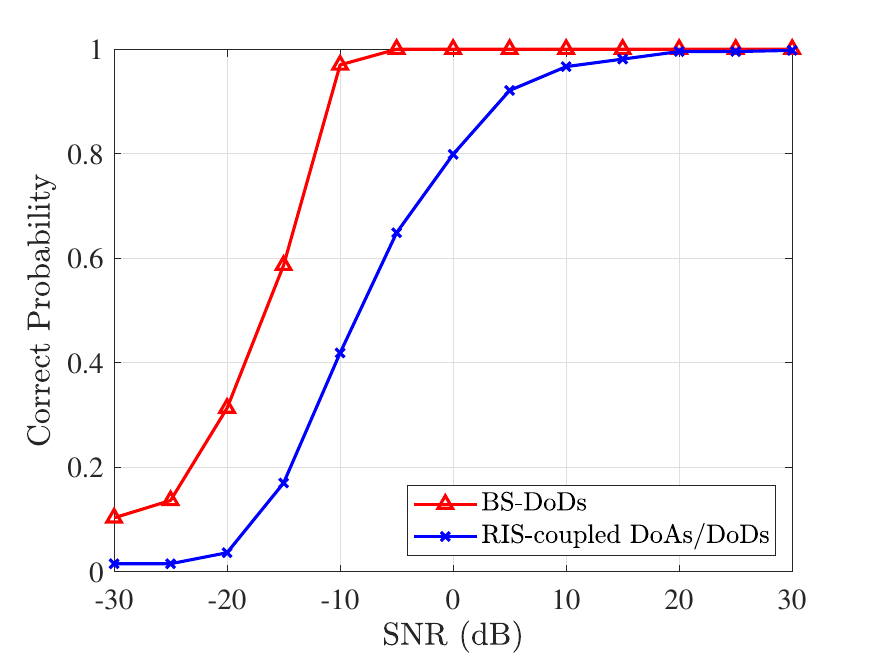}
	\caption{Correct probability of DoDs at BS and coupled DoAs/DoDs at RIS against SNR.}
	\label{fig:Pro_SNR}
\end{figure}

Fig. \ref{fig:Pro_SNR} depicts the correct probability performance of the DoDs at BS and coupled DoAs/DoDs at RIS versus SNR with $T=50$ and $P=30$. It is shown that both the correct probabilities approach one as SNR increases. This indicates that the proposed scheme can ensure a certain level of reliability even under low SNR conditions and has the ability to obtain accurate physical directions with sufficient channel conditions. Additionally, the correct probability of the coupled directions at RIS is consistently lower than that of the DoDs at BS across all SNR values. This is because the estimation of the coupled directions at RIS depends on the accuracy of the DoDs estimation at BS. Angular estimation errors at BS would propagate to RIS, leading to performance gap between the two probability curves. This discrepancy highlights the propagation of errors from Phase I to Phase II: the angular estimation errors at BS directly impact the estimation accuracy of the coupled directions at RIS, which in turn affects the CSI estimation for other SCs in subsequent steps. However, as the SNR increases, the influence of these propagated errors diminishes, which enables Phase II to achieve improved CSI recovery performance.

\begin{figure}[!t]
	\centering
	\includegraphics[width=0.48\textwidth]{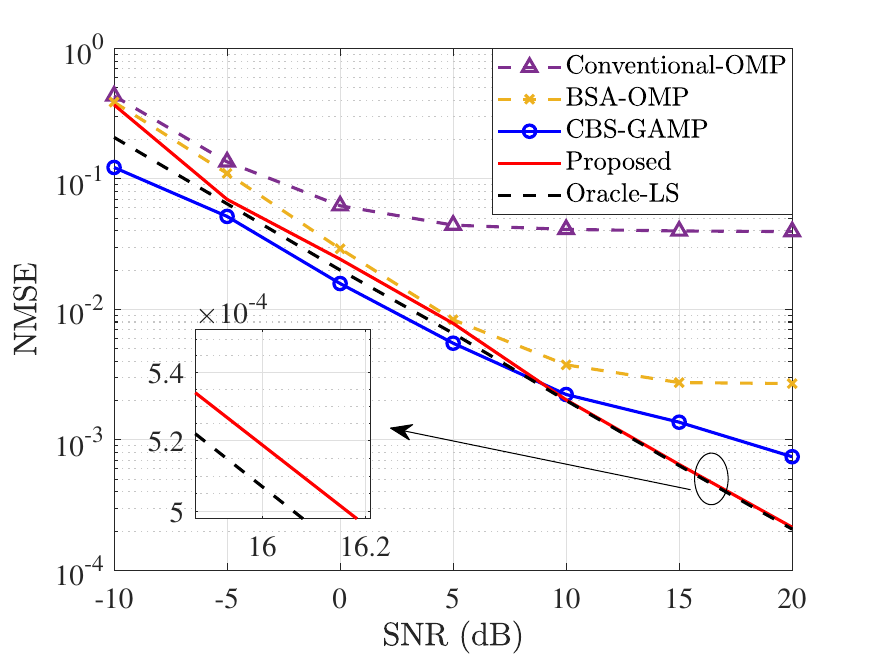}
	\caption{NMSE performance comparison against SNR.}
	\label{fig:NMSE_SNR}
\end{figure}
Fig. \ref{fig:NMSE_SNR} illustrates the NMSE performance versus SNR with $T=50$ and $P=30$. Overall, it can be observed that the proposed scheme outperforms conventional OMP and BSA-OMP algorithm across all SNR regions. This is because the proposed scheme has the ability to exploit unknown channel information, and the proposed dictionaries can effectively cover the beam space under beam split effect. Additionally, the proposed scheme is slightly inferior to that of CBS-GAMP algorithm when SNR is less than $10$ dB. This is because estimation errors generated by partial SCs propagate more significantly to the remaining SCs under poor channel conditions. However, the proposed scheme outperforms CBS-GAMP algorithm in high SNR regions, i.e., SNR is larger than $10$ dB. This is because angle decoupling provides more accurate angular estimation information. This suggests that there exists a performance-complexity trade-off between these two algorithms. \change{Moreover, it is observed that the proposed scheme progressively  approaches the NMSE performance of oracle-LS benchmark as SNR increases.
Under high SNR conditions, the performance gap between the two schemes becomes negligible, further verifying the effectiveness of the proposed scheme.}

\begin{figure}[!t]
	\centering
	\includegraphics[width=0.48\textwidth]{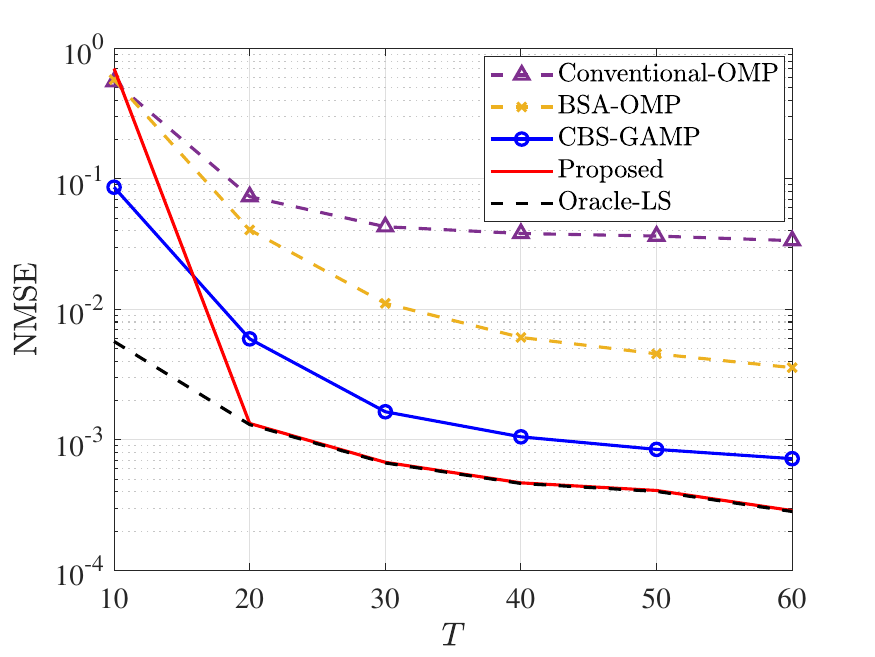}
	\caption{NMSE performance comparison against the number of subframes $T$.}
	\label{fig:NMSE_T}
\end{figure}
Fig. \ref{fig:NMSE_T} shows the NMSE performance versus the number of subframes $T$, with SNR~$= 20$~dB and $P = 15$. As expected, the NMSE performance of all the algorithms improves with the increase in the number of subframes $T$. It can be observed that the proposed scheme maintains good performance even with a short length of subframes, e.g., $T=20$, demonstrating the robustness of the proposed algorithm. Moreover, compared to the CBS-GAMP method, the proposed algorithm achieves similar performance with approximately a $50\%$ subframes reduction. This indicates that the proposed scheme can efficiently reduce the pilot overhead for channel estimation.

\begin{figure}[!t]
	\centering
	\includegraphics[width=0.48\textwidth]{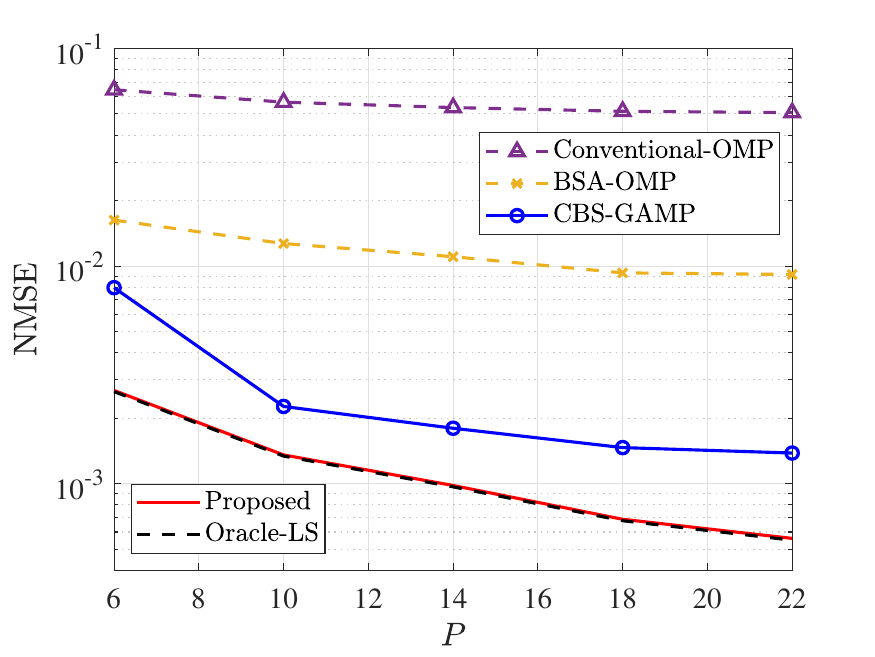}
	\caption{NMSE performance comparison against the number of time slots $P$.}
	\label{fig:NMSE_P}
\end{figure}
Fig. \ref{fig:NMSE_P} provides the NMSE performance versus the number of time slots $P$, with SNR~$= 20$~dB and $T = 30$. As expected, the NMSE performance of all the algorithms improves with the increase in the number of slots $P$. It can be observed that the proposed scheme outperforms all other three algorithms and can achieve near-optimal NMSE performance of the ideal oracle-LS upper bound in all considered values of $P$. The performance gap between the proposed scheme and other existing schemes is quite large. This indicates that the proposed scheme can effectively reduce the system complexity for channel estimation.

\begin{figure}[!t]
	\centering
	\includegraphics[width=0.48\textwidth]{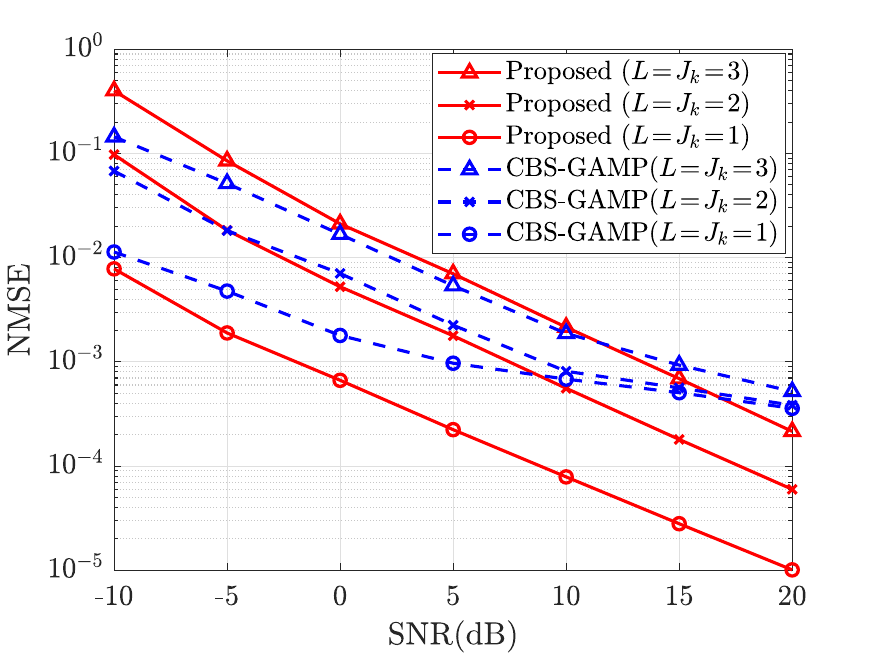}
	\caption{NMSE performance comparison against the number of paths $L,J_k$.}
	\label{fig:NMSE_L}
\end{figure}
Fig. \ref{fig:NMSE_L} illustrates the NMSE performance versus the number of paths $L$ and $J_k$, with $T = 50$ and $P = 30$. It can be observed that as the the number of paths increases, the NMSE performance of both the two schemes degrade. This is because of the increased estimation parameters and challenges in distinguishing directions among more multipath components. Furthermore, as SNR increases, the NMSE of CBS-GAMP algorithm converges to around $5\times10^{-4}$ regardless of the number of paths. In contrast, the NMSE of the proposed scheme remain decreasing across all SNR regions. This demonstrates that the proposed scheme maintains better estimation accuracy and is more robust against variations in the number of paths compared to the CBS-GAMP algorithm.

\begin{figure}[!t]
	\centering
	\includegraphics[width=0.48\textwidth]{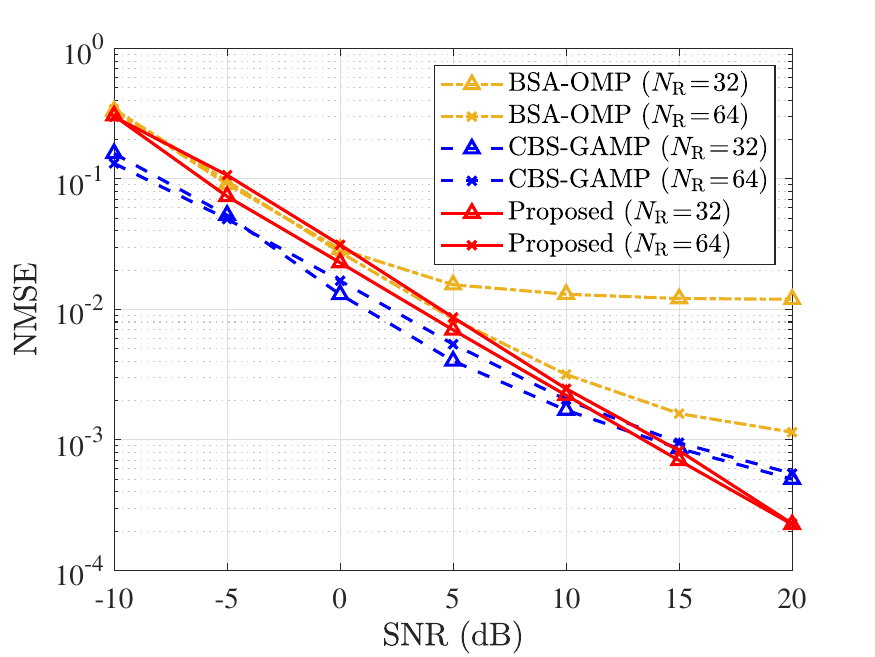}
	\caption{NMSE performance comparison against the number of reflecting elements of RIS $N_{\mathrm{R}}$.}
	\label{fig:NMSE_NRIS}
\end{figure}
Fig. \ref{fig:NMSE_NRIS} shows the NMSE performance versus the number of reflecting elements of RIS $N_{\mathrm{R}}$, with $T = 50$ and $P = 30$. It can be observed that as $N_{\mathrm{R}}$ decreases, the NMSE performance of the BSA-OMP algorithm degrades significantly. This is because a small-scale RIS cannot provide enough measurement data diversity to support OMP based algorithms. In contrast, the performance of the proposed and CBS-GAMP schemes slightly improve as $N_{\mathrm{R}}$ decreases. This indicates that these two schemes exhibit robust framework, which are less susceptible to the lack of data diversity.


\section{Conclusion}\label{sec:Conclusion}
In this paper, we investigated channel estimation in RIS-assisted wideband THz systems with consideration of beam split effect. We proposed a two-phase channel estimation scheme to achieve accurate CSI acquisition with reduced computational complexity. Specifically, in Phase I, we introduced a two-step method to achieve full CSI estimation over two SCs. In Phase II, we employed angular information obtained from Phase I to estimate the channels for the remaining SCs using a simple LS algorithm. Theoretical analysis and numerical results demonstrated that our scheme achieved near-oracle performance in terms of NMSE. Moreover, compared to existing solutions, the proposed two-phase strategy significantly reduced computational complexity while maintaining a high level of accuracy in CSI acquisition.

\appendices
\section{Proof of Proposition 1}
\label{appendix:proof_1}

We first define the parameter sequence $\bar{\vartheta}_{q_1}=\delta \cdot q_1-(Q_{\mathrm R}+1)/{Q_{\mathrm R}} , \forall q_1= 1,\dots,Q_{\mathrm R}$ which corresponds one-to-one with $Q_{\mathrm R}$ ARVs in the dictionary $\mathbf{C}_{\mathrm{R},m}$. Similarly, we define the sequence $\delta_{q_2}=-\delta\cdot q_2 +(Q_{\mathrm R}+1)/Q_{\mathrm R}$ to correspond to $Q_R$ ARVs in the conjugate dictionary $\mathbf{C}_{\mathrm{R},m}^{*}$. Thus, the column vectors of $\mathbf{C}_{\mathrm{R},m}^{*}\bullet \mathbf{C}_{\mathrm{R},m}$ can be expressed as $\mathbf{a}_{\mathrm{R},m}(\bar{\vartheta}_{q_1}-\bar{\vartheta}_{q_2})$ for $q_1,q_2 = 1,\dots,Q_{\mathrm{R}}$. The $j$-th column of the coupled dictionary $\tilde{\boldsymbol{\xi}}_{\mathrm{R},m,j}$ can be expressed as
	\begin{equation}\label{equ:xi}
		\tilde{\boldsymbol{\xi}}_{\mathrm{R},m,j}=\mathbf{a}_{\mathrm{R},m}(\delta \left( q_1-q_2 \right) ),
	\end{equation}
	where $j=\left( q_2-1 \right) Q_{\mathrm{R}}+q_1,\forall q_1,q_2 \in 1,\cdots,Q_{\mathrm{R}}$. Based on the fact that $q_1-q_2\in\left[1- Q_{\mathrm{R}},Q_{\mathrm{R}}-1\right] $, the number of unique columns of coupled dictionary $\tilde{\mathbf{\Xi}}_{\mathrm{R},m}$ is $Q_\mathrm{B}= 2Q_{\mathrm{R}}-1$. 
    \begin{figure}[htbp]
	\centering
	\includegraphics[width=0.48\textwidth]{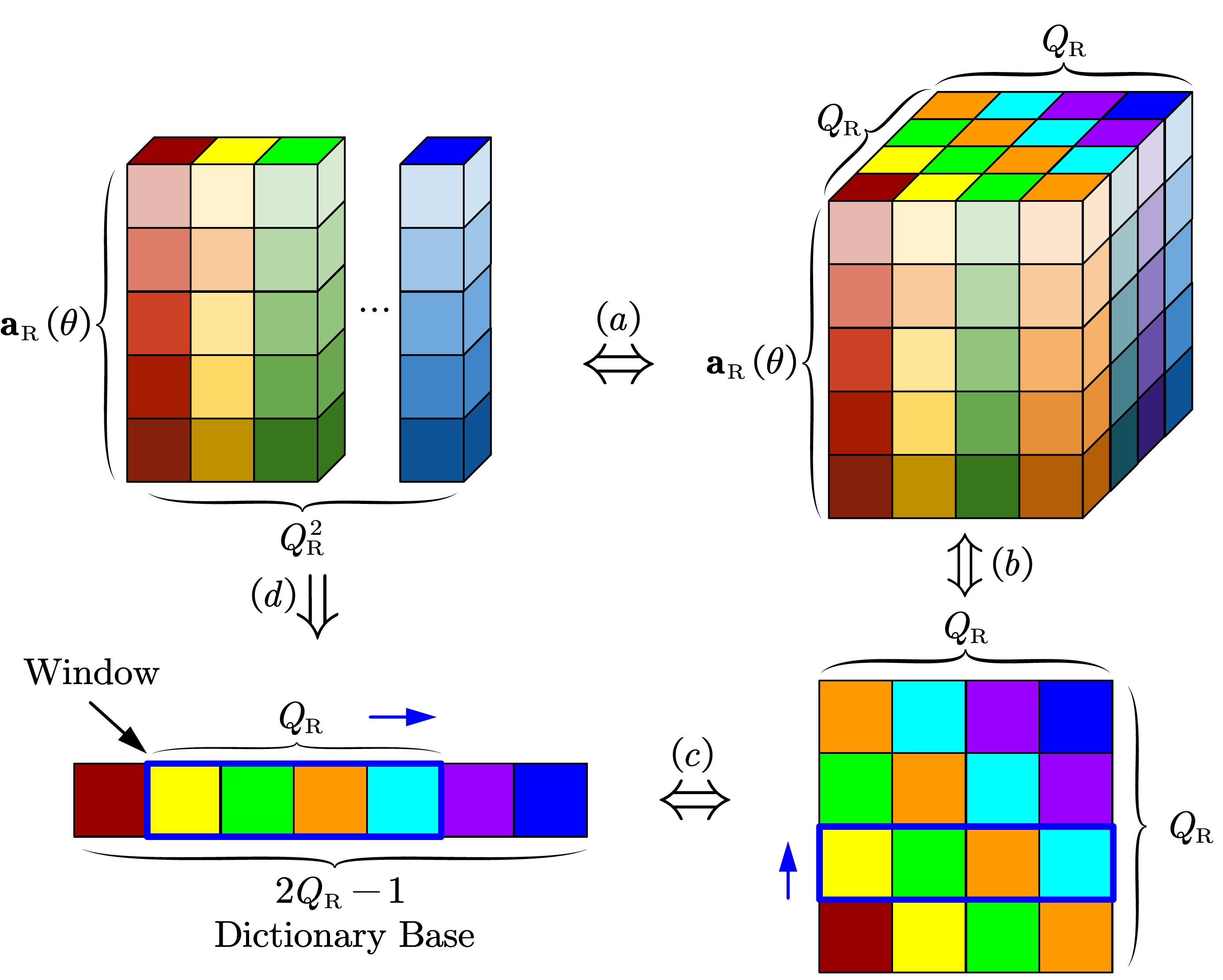}
	\caption{Coupled dictionary simplification diagram.}
	\label{fig:simply_dic}
\end{figure}
	
	Then, we simplify the coupled dictionary $\tilde{\mathbf{\Xi}}_{\mathrm{R},m}$ step by step from (a) to (d) as shown in Fig. \ref{fig:simply_dic}. As depicted in step (a), we reshape the overcomplete dictionary $\tilde{\mathbf{\Xi}}_{\mathrm{R},m}$ of length $Q_{\mathrm{R}}^2$ ARVs into a three-dimensional matrix where each column represents an ARV. Given that ARV can be uniquely determined by known angles in \eqref{equ:xi}, we use the ordered sequence $\mathcal{L}_Q = \left(1,\dots,Q_{\mathrm{R}}\right)$ to index the corresponding ARVs. In step (b), we form the index matrix $\mathbf{Q} \in \mathbb{N}^{Q_\mathrm{R}\times Q_\mathrm{R}}$ with elements $[\mathbf{Q}]_{i,j} = (\mathcal{L}_Q)_i - (\mathcal{L}_Q)_j$ for $i,j=1,\dots,Q_{\mathrm{R}}$. It is noticed that due to the cyclic shift property of this dictionary, the dictionary base can be represented entirely by combining just the first and last rows. This means that we can reconstruct the original dictionary $\tilde{\mathbf{\Xi}}_{\mathrm{R},m}$ with a sliding window of length $Q_{\mathrm{R}}$. Therefore, the simplified CBS dictionary $\mathbf{\Xi }_{\mathrm{R},m}$ can be given as in \eqref{equ:Xi_simp}.
	Moreover, substituting $q_2=Q_{\mathrm{R}}$ into \eqref{equ:xi}, the $q$-th column where $q = 1,\dots Q_\mathrm{B}$ of $\mathbf{\Xi}_{\mathrm{R},m}$ can be formulated~as
	\begin{equation}
		\boldsymbol{\xi} _{\mathrm{R},m,q}=\mathbf{a}_{\mathrm{R},m}\left( 2\frac{q-Q_{\mathrm{R}}}{Q_{\mathrm{R}}} \right).
	\end{equation}


\bibliographystyle{IEEEtran}
\bibliography{./bibtex/IEEEabrv,./bibtex/IEEEreference}


\end{document}